%% file: Paper_TW_Feb_15_0254_R2.tex
\documentclass[journal,twocolumn,twoside]{IEEEtran}
\pdfoutput=1
\usepackage{amsmath}
\usepackage{amssymb}
\usepackage{amsthm}
\usepackage{bbm}\let\mathbb\mathbbmss
\usepackage{graphicx}\graphicspath{{./figs/}}
\usepackage[caption=false]{subfig}
\usepackage{cite}
\usepackage{booktabs}
\usepackage{microtype}
\usepackage[hyphens]{url}

\providecommand{\texorpdfstring}[2]{#1}

\theoremstyle{plain}
\newtheorem{lemma}{Lemma}
\newtheorem{theorem}{Theorem}
\newtheorem{proposition}{Proposition}
\newtheorem{corollary}{Corollary}[theorem]
\theoremstyle{definition}
\newtheorem{definition}{Definition}

\newtheorem{assumption}{Assumption}
\theoremstyle{remark}

\input{notation}

\allowdisplaybreaks
\newlength{\figurewidth}
\begin{document}

\title{Towards a Tractable Analysis of Localization\\Fundamentals in Cellular Networks}

\author{Javier Schloemann, {\em Member, IEEE,} Harpreet S. Dhillon, {\em Member, IEEE,} and 
\\R. Michael Buehrer, {\em Senior Member, IEEE} %
%\thanks{Manuscript received February 24, 2015; revised July 12, 2015 and September 30, 2015; accepted October 5, 2015. The associate editor coordinating the review of this paper and approving it for publication was H.-C. Wu.}
\thanks{J. Schloemann is with Northrop Grumman Corporation, Raleigh, NC. He was with Wireless@VT, Department of ECE, Virginia Tech, Blacksburg, VA, USA. Email: javier@vt.edu. H. S. Dhillon and R. M. Buehrer are with Wireless@VT, Department of ECE, Virginia Tech, Blacksburg, VA, USA. Email: \{hdhillon, buehrer\}@vt.edu. This paper was presented in part at the IEEE ICC 2015 Workshop on Advances in Network Localization and Navigation (ANLN), London, UK~\cite{Schloemann2015a}. \hfill Manuscript last updated: \today.}
}

\maketitle

\begin{abstract}
When dedicated positioning systems, such as GPS, are unavailable, a mobile device has no choice but to fall back on its cellular network for localization. Due to random variations in the channel conditions to its surrounding base stations (BS), the mobile device is likely to face a mix of both favorable and unfavorable geometries for localization. Analytical studies of localization performance (e.g., using the Cram\'{e}r-Rao lower bound) usually require that one fix the BS geometry, and favorable geometries have always been the preferred choice in the literature. However, not only are the resulting analytical results constrained to the selected geometry, this practice is likely to lead to overly-optimistic expectations of \emph{typical} localization performance. Ideally, localization performance should be studied across all possible geometric setups, thereby also removing any selection bias. This, however, is known to be hard and has been carried out only in simulation. In this paper, we develop a new tractable approach where we endow the BS locations with a distribution by modeling them as a Poisson point process (PPP), and use tools from stochastic geometry to obtain easy-to-use expressions for key performance metrics. In particular, we focus on the probability of detecting some minimum number of BSs, which is shown to be closely coupled with a network operator's ability to obtain satisfactory localization performance (e.g., meet FCC E911 requirements). This metric is indifferent to the localization technique (e.g., TOA, TDOA, AOA, or hybrids thereof), though different techniques will presumably lead to different BS hearability requirements. In order to mitigate excessive interference due to the presence of \emph{dominant} interferers in the form of other BSs, we incorporate both BS coordination and frequency reuse in the proposed framework and quantify the resulting performance gains analytically.
\end{abstract}


%%%%%%%%%
% Introduction %
%%%%%%%%%

\section{Introduction}\label{Sec:Introduction}

\IEEEPARstart{G}{eolocation} (also called \emph{positioning}, \emph{localization}, and \emph{position location}) is deeply ingrained in our daily lives and has been studied by the scientific community for many years~\cite{Stein1981,Torrieri1984,Stein1993,Chan1994,Zekavat2012}. The driving force behind much of the research is a mandate by the Federal Communications Commission (FCC) requiring cellular network operators to locate those calling 911 to within certain accuracy requirements~\cite{Rappaport1996,Reed1998}. Until recently, these requirements included only outdoor location accuracies. Accordingly, the predominant way cellular network operators have met the requirements of the mandate is by relying on the Global Positioning System (GPS). In January of 2015, however, the FCC expanded its mandate to include a phase-in of indoor positioning requirements, citing that the bulk of emergency calls now originate indoors~\cite{FCC2015}.

While accurate outdoor positioning using GPS is reliably available under clear sky conditions, many years of positioning study have not yet resulted in equally reliable positioning methods in \emph{urban canyons} and indoor scenarios. Consider the classical example of locating emergency personnel (e.g., firefighters) indoors, where our current inability to provide accurate indoor positioning has recently been described as a dilemma ``where people [are] literally dying within a hundred feet of safety''~\cite{Harris2013}. Despite this, economic limitations dictate that global navigation satellite systems (GNSS), such as GPS, are likely to be the only widespread dedicated location systems in the foreseeable future. This necessitates a fallback to terrestrial cellular networks for geolocation in  situations where these prevalent location technologies are unavailable. Presently, however, no analytical approaches exist to study the fundamentals of localization performance in these networks.

It is the objective of this paper to introduce a new tractable model for studying terrestrial geolocation using cellular networks. The model uses concepts from point process theory~\cite{Kingman1993} and stochastic geometry~\cite{Haenggi2013} to lend tractability to the study of network localization, which is demonstrated through accurate and easy-to-use expressions characterizing the number of base stations (BSs) that are able to participate in a localization procedure, a key factor influencing the localization performance.

\subsection{Prior art and motivation}
Regardless of the technique, geolocation performance fundamentally depends upon three things: \emph{(i) the number of participating BSs, (ii) the geometry or locations of these BSs relative to the device being localized, and (iii) the accuracy of the positioning observations}. When these factors are deterministic, localization performance is typically studied analytically using the Cram\'{e}r-Rao lower bound (CRLB) (e.g.~\cite{Chang2004,Savvides2005,Guvenc2009}). Furthermore, since localization systems are being discussed, the topologies considered are often ones favorable to location estimation (such as by using a carefully controlled set of scenarios to avoid poor geometries, e.g.~\cite{Savvides2005}). This approach, however, is not appropriate for studying cellular localization performance because none of the aforementioned factors may reasonably be considered deterministic. Instead, the different possible locations of the mobile device within a network, the topology of the BSs relative to the mobile device, and random variations in the channel conditions will (i) affect the number of BSs capable of participating in the localization procedure, (ii) result in those BSs providing a mix of both good and bad topologies for localization, and (iii) define the strengths of the received BS signals which directly impact the accuracies of the positioning observations. Clearly, only a very large set of deterministic scenarios could provide insight into localization performance over all possible variations of the above. Even so, crafting realistic scenarios requires the assumption of specific system design parameters and propagation conditions, meaning that any insights from the results would be specific to the chosen assumptions, rather than being general. This makes a classic metric such as the CRLB not suitable for a general analysis.

In this work, our focus is on providing a truly general analysis of one metric which we will see is indicative of a network operator's ability to meet its mandated localization accuracy requirements--the number of successfully detectable BSs. For communication, it is ideal for a mobile device to receive a strong signal from its serving BS and weak signals from all its neighbors. For localization, on the other hand, it is increasingly beneficial for a mobile device to receive usable signals from more and more neighbors. Although communication is the primary concern of network operators, cellular system designers have no choice but to also cater to localization demands. As a result, the need to satisfy these conflicting goals has a name among designers--the \emph{hearability problem}, due to the fact that increasing hearability from neighboring cells for the purposes of positioning is contrary to the principles of cellular system design~\cite{R1-090053}. In order to provide a general analysis of cellular hearability, it is appropriate to characterize this metric across all possible network topologies and channel conditions for the aforementioned reasons.

One field of study in which random network topologies are considered when discussing localization is that of wireless or ad-hoc sensor networks. However, these studies differ from the cellular scenario we consider in that the sensor network literature almost certainly ignores interference and propagation effects, instead opting to define coverage using some fixed detection range~\cite{Savvides2005,Li2008,Gribben2010}, though in rare cases the received signal strength may be used to define a circular coverage region when shadowing is also considered~\cite{Daneshgaran2007}. In order to apply to cellular networks, our analysis differs by not only including interference, but also by employing a model enriched with additional cellular and localization-specific parameters.

Lastly, when cellular networks are specifically considered, they are typically studied using the popular hexagonal grid model, making analytical results difficult to come by. Instead, network designers resort to complex system-level simulations using a common set of agreed-upon evaluation parameters~\cite{R1-091443}, as is clear from the 3GPP standardization process~\cite{R1-091912}. This can also be seen in the limited amount of literature on hearability~\cite{Yap2002,Oborina2009,Vaghefi2014a}. Though simulation-based, the work in \cite{Yap2002} and \cite{Oborina2009} is most similar to ours in that the focus is on studying cellular hearability (in GSM and WCDMA~\cite{Yap2002} and LTE~\cite{Oborina2009}). Interference, coordination, and fractional load are considered, but in the end, one is left with only simulation results, from which it is not possible to gather general insights into how localization performance is impacted by system design parameters and propagation effects. This motivates the need for a tractable analytical model that can provide preliminary design insights and will either circumvent the need for simulations completely or limit the ranges of the simulation parameters. In~\cite{Andrews2011} and \cite{Dhillon2012}, such tractable, yet accurate, analytical models were developed to study the coverage and rate of single-tier and multi-tier cellular communication systems, respectively. It was shown that the spatial layout of the BSs in cellular networks can be tractably and realistically modeled using Poisson point processes (PPPs), especially as cellular networks continue to deviate from centrally-planned macro-cell networks to networks which include an increasing number of more arbitrary small-cell deployments such as picocells and femtocells. This opens the door to the use of powerful tools from stochastic geometry to derive closed form expressions for key performance metrics. Motivated by these advances, we develop a similar approach to lend tractability to the analysis of localization systems.

\subsection{Contributions}
\noindent The main contributions of this paper are as follows.

\noindent {\em A general model for studying localization in cellular networks}: In Section~\ref{Sec:SystemModel}, we build on the foundations of \cite{Andrews2011} and \cite{Dhillon2012} and present a tractable model specifically for studying localization in cellular networks. The model includes average network load, the ability for BSs participating in the localization procedure to coordinate transmissions (e.g., through \emph{joint scheduling} in LTE~\cite{R1-090053}), and is easily adapted to frequency reuse when frequency bands are assigned randomly. From \cite{Andrews2011}, it is easy to see that self-interference from the network often leads to poor coverage probabilities even when the reception of only a single signal is required. Since localization procedures require the successful reception of multiple signals and are thus even more demanding, localization systems typically necessitate the integration of signals over time in order to improve detectability. Our model assumes that this integration provides some fixed processing gain, while fast fading is excluded under the assumption that it is averaged out at the receiver. This exclusion of small-scale fading effects in our model agrees with the state-of-the-art 3GPP simulation models~\cite{R1-091443}.

\noindent {\em Definition and characterization of accuracy-related hearability metric}: We define a metric in order to study the number of BSs whose signals arrive with sufficient quality to successfully participate in a localization procedure. After presenting a clear relationship between the metric and positioning accuracy, upper bounds and approximations for the distribution of this metric are derived using the proposed model, resulting in expressions that are closed-form in some special cases of interest, while easy-to-compute integral expressions for the others. A dominant-interferer type analysis is employed~\cite{Heath2013,Weber2007}, whereby the strongest interference source is always treated exactly, leading to remarkably accurate approximations, nearly indistinguishable from truth.

\noindent {\em Design insights}: Lastly, we present results and draw out some interesting observations useful in the design of localization systems. We make evident the fact that localization performance is limited by the interference from the BSs participating in the localization procedure. It is shown that mitigating this interference through BS cooperation helps up to a point, beyond which the true gains in localization performance are realized through frequency reuse. The need to employ frequency reuse for accurate positioning has been observed previously, but only by first going through the process of running many complex system-level simulations. In LTE, for example, a reuse of six was deemed necessary for downlink TDOA positioning, a verdict reached after many simulations \cite{R1-091912}.

%%%%%%%%%%%
% Proposed Model %
%%%%%%%%%%%
\section{System Model}\label{Sec:SystemModel}
We now formally describe the proposed model. The key notation presented in this section is summarized in Table~\ref{Table:Notation}.

\begin{table}
\centering
\caption{Summary of Key Notation}
\label{Table:Notation}
\begin{tabular}{@{}c@{\hspace{2em}}l@{}}
\toprule
\textbf{Notation} & \textbf{Description} \\
\midrule
%$\alpha$ & Path loss exponent ($\alpha > 2$) \\
$\| \cdot \|$; $\alpha$ & $\ell_2$-norm; path loss exponent ($\alpha > 2$) \\
$\PPPa$; $\pppai$ & PPP of BS locations; Density of $\PPPa$ \\
$\origin$ & Origin (location of the typical user) \\
%$\pppai$ & Density of $\PPPa$ \\
$x_i$ ($\xa_i$) & Location of the $i$th closest (active) BS in $\PPPa$\\
$R_i$ ($\Ra_i$) & Distance of $x_i$ ($\xa_i$) from the origin \\
$\pg$ & Processing gain \\
$\threshold$ & Post-processing SINR required for a successful wireless link \\
$L$ & Number of participating BSs \\
$p$ & Activity factor among $L$ participating BSs \\
$q$ & Average network load for non-participating BSs \\
$\Omega$ & Number of participating BSs interfering with the $L\th$ BS \\
$\b(\theta, r)$ & A ball centered at $\theta$ with radius $r$ \\
$\vert \ncalA \vert$; $\ncalA^c$ & The Lebesgue measure of region $\ncalA$; the complement of $\ncalA$ \\
%$\ncalA^c$ & The complement of region $\ncalA$ \\
$\ncalB \subseteq \ncalA$ & Region $\ncalB$ is contained in $\ncalA$ \\
$\ncalA \backslash \ncalB$ & Region $\ncalA$ excluding its overlap with $\ncalB$; $\ncalA \backslash \ncalB = \ncalA \cap \ncalB^c$ \\
$\indicator(\cdot)$ & Indicator function; % $\indicator(\cdot) \in \{0,1\}$; 
$\indicator(\ncalA) =1$ if $\ncalA$ is true, $0$ otherwise\\
\bottomrule
\end{tabular}
\end{table}

\subsection{Spatial base station layout}
The locations of the BSs are modeled using a homogeneous PPP $\PPPa \in \R^2$ with density $\pppai$~\cite{Haenggi2013}. Due to the stationarity of a homogeneous PPP, the device to be localized is assumed to be located at the origin $\origin$. If the interference is treated as noise at the receiver, the most appropriate metric that captures link quality is the signal-to-interference-plus-noise ratio (SINR). For the link from some BS $x \in \PPPa$ to the origin, the SINR can be expressed as:
\begin{equation}
\SINR_x = \frac{P \ncalS_{x} \|x\|^{-\alpha}}{\sum_{\substack{y \in \PPPa\\y \neq x}} P \ncalS_{y} \|y\|^{-\alpha} + \sigma^2},
\label{Eq:SINR_k}
\end{equation}
where $P$ is the transmit power, $\ncalS_z$ denotes the independent shadowing affecting the signal from BS $z$ to the origin, $\alpha > 2$ is the pathloss exponent, and $\N$ is the noise variance. Note that \eqref{Eq:SINR_k} represents the SINR prior to any processing gain, yet as will be evident in the sequel, positioning systems typically have to work at lower target SINRs, thereby necessitating the need for some form of processing gain. In general, the post-processing SINR will include some multiplicative factor $\pg$ representing the processing gain, which depends upon system parameters (e.g. integration time) and is assumed to average out the effect of small scale fading. This is the reason why the SINR expression in \eqref{Eq:SINR_k} does not contain a fast fading term, which is consistent with current models for evaluating cellular positioning performance~\cite{R1-091443}. Those conversant with stochastic geometry-based analyses of wireless networks will recognize that the absence of fast fading on the serving link, specifically one from the exponential family of distributions (e.g., Rayleigh fading), adds to the technical challenge of a localization system analysis. Lastly, although we consider interference as noise, the gain from some multiuser detection technique, such as interference cancellation, can be abstracted into the proposed model using the processing gain parameter $\pg$~\cite{Al-Tameemi2014}.

\subsection{Selection of participating base stations}\label{Sec:BSSelection}

As is common in communication system analyses, serving BSs are selected according to the \emph{strongest BS association policy}, measured using average signal strength, which typically includes long time-scale effects such as shadowing and pathloss. Now, consider that we desire to use a total of $L$ BSs for positioning. Thus, we assume that the $L$ BSs which provide the highest average received power make up the set of \emph{participating BSs}. Their successful participation, however, is not guaranteed as there is some post-processing SINR threshold $\threshold$ (or equivalently, pre-processing SINR threshold $\threshold/\pg$) above which the signals from the participating BSs must arrive in order for them to successfully contribute to the localization procedure. In the absence of shadowing, the set of potential BSs simply corresponds to the set of the $L$ nearest BSs. When shadowing is considered and BSs are selected according to average signal strength, the effect of shadowing may be absorbed as a perturbation in the locations of the BSs provided that the fractional moment $\E\left[\ncalS_z^{2/\alpha} \right] < \infty$~\cite{Blaszczyszyn2013a,Dhillon2014}. Thus, when this condition is fulfilled and without loss of generality, we define a new equivalent PPP with density $\pppai\,\E\left[\ncalS_z^{2/\alpha} \right]$. In doing so, we ensure that the strongest BS association policy in the original PPP is equivalent to the \emph{nearest BS association policy} in the transformed PPP. For notational simplicity, we will continue to represent the transformed PPP by $\PPPa$ with density $\pppai$, under the assumption that if shadowing is present, it is already reflected in the density of $\PPPa$. Note that the condition on the fractional moment is fairly mild and will almost always be satisfied by the distributions of interest, including the most common assumption of log-normal shadowing with finite mean and standard deviation~\cite{Dhillon2014}. As a final note, it has been shown that for SINR, shadowing causes even more regular network models, such as the common hexagonal lattice, to behave like a PPP model~\cite{Blaszczyszyn2013,Keeler2014,Blaszczyszyn2015}. This further validates the use of a PPP to model BS locations.

\subsection{Base station coordination and network load}
BS coordination and network load are modeled through two activity factors, $p$ and $q$, respectively. During a localization procedure, the $L$ participating BSs coordinate with each other and attempt to blank their own transmissions while the others are active, but are unable to do so with probability $p$ due to network traffic demands. The remaining BSs are assumed active with probability $q$, which is simply the average network load~\cite{Dhillon2013}. For simplicity of exposition, let us order the BSs in the now shadowing-transformed $\Phi$ in terms of increasing distance from the origin such that the location of the $k^\text{th}$ farthest BS from the origin is denoted by $x_k \in \Phi$. We now enrich the previous SINR expression in \eqref{Eq:SINR_k} for the signals arriving from the participating BSs (i.e., $x_k$ for $k \in \{1, \ldots, L\}$) by including BS coordination and network load as follows:
\begin{equation}
\SINR_{k}(L) = \frac{P \|x_k\|^{-\alpha}}{\displaystyle \sum_{\substack{i=1\\i \neq k}}^L a_i P \|x_i\|^{-\alpha} \! + \!\!\!\sum_{j=L+1}^\infty b_j P \|x_j\|^{-\alpha} + \sigma^2},
\label{Eq:SINR_k_L}
\end{equation}
where $a_i$ and $b_j$ are independent indicator random variables (fixed throughout the localization procedure) modeling the thinning of the interference field by taking on value 1 with probabilities $p$ and $q$, respectively. Note that $\SINR_k$ is now a function of $L$ because of the potentially different activity factors for the participating BSs and the rest of the network. The number of active participating BSs interfering with the $L^\text{th}$ BS is $\Omega = \sum_{i=1}^{L-1} a_i$, which is a binomial random variable with probability mass function
\begin{equation}
\pomega = \binom{L-1}{\omega} p^\omega (1-p)^{L-1-\omega},\qquad \omega \in \{0, \ldots, L-1\}.
\end{equation}

\subsection{Base station range distributions}

\begin{figure}
\centering
\includegraphics[width=.8\columnwidth]{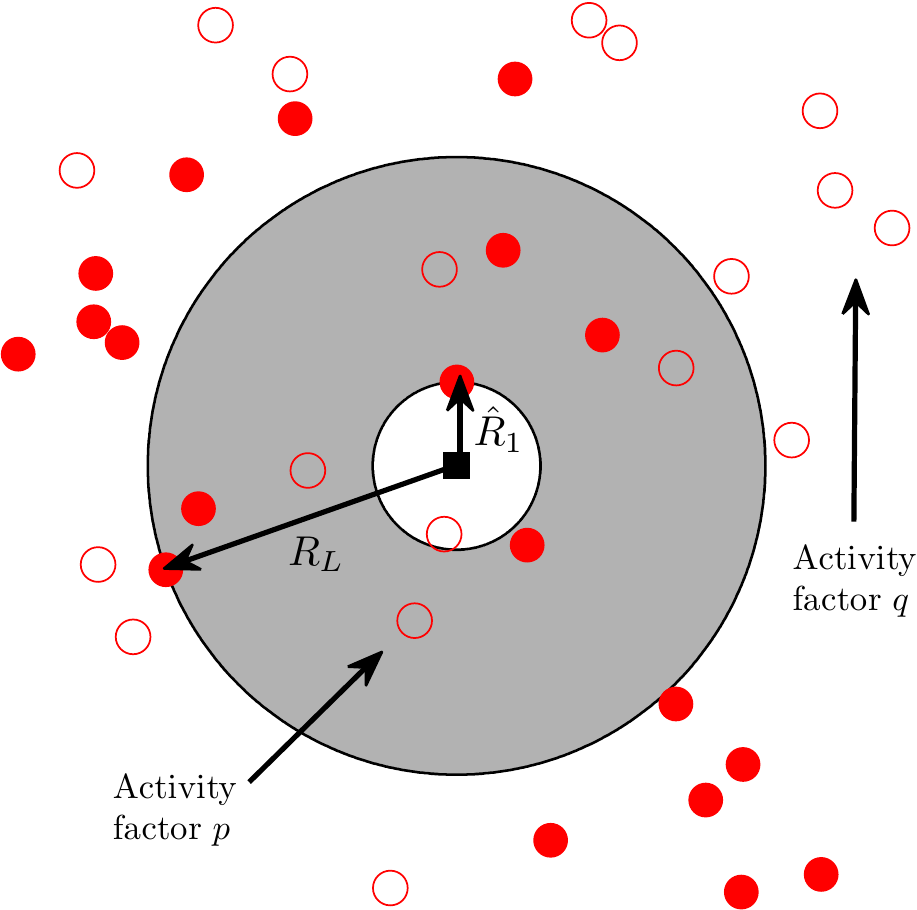}
\caption{\textsc{The proposed model}. The mobile device is denoted by the filled square and actively transmitting BSs are denoted by solid circles.}
\label{Fig:Model_Example}
\end{figure}

From \eqref{Eq:SINR_k_L}, it is clear that the SINRs are dependent on the distances of the BSs from the origin, rather than the locations themselves. Thus, it is worthwhile to characterize these. Let $R_k = \|x_k\|$ and $\Ra_m = \|\xa_m\|$, where $\xa_m$ is now the $m^\text{th}$ farthest \emph{active} BS. This is illustrated in Figure~\ref{Fig:Model_Example} for the closest active ($\Ra_1$) and $L^\text{th}$ closest overall ($R_L$) BSs. The distribution of $R_L$ is~\cite{Haenggi2005}:
\begin{align}
	f_{R_L}(r)
	&= e^{-\pppai \pi r^2 } \frac{2 (\pppai \pi r^2)^L}{r \GammaL}.
\end{align}
Now, conditioned on $\Omega$ and the distance of the $L^\text{th}$ BS from origin, we present a useful lemma for understanding the distribution of the locations of the $\Omega$ active participating BSs. 

\begin{lemma}\label{Lemma:CircularBPP}
Conditioned on the distance of the $L^\text{th}$ BS from the origin, $R_L$, the $\Omega$ active BSs closer to the origin than the $L^\text{th}$ BS are independent and identically distributed over $\b(\origin, R_L)$ with the location of each BS sampled uniformly at random from $\b(\origin, R_L)$.
%distributed according to a binomial point process (BPP) (i.e., in a uniformly random manner) inside the circle of radius $R_L$ centered at the origin.
\end{lemma}
\begin{proof}
See Appendix~\ref{Proof:AnnularBPP}.
\end{proof}

Note that a point process is which a fixed number of points are distributed uniformly at random independent of each other in a given compact set is termed as a {\em uniform Binomial Point Process (BPP)}~\cite{Haenggi2013}. Therefore, the above Lemma simply states that conditioned on $R_L$, $\Omega$ active BSs closer to the origin than the $L^\text{th}$ BS form a BPP on $\b(\origin, R_L)$.

\noindent Using Lemma~\ref{Lemma:CircularBPP}, we obtain the following distribution for the most dominant (i.e., closest active) interferer, given the distance of the $L^\text{th}$ BS and $\Omega$. This is formally presented below.
\begin{lemma}\label{Lemma:FRa1}
The cumulative distribution function of the closest active BS distance $\Ra_1$ given $R_L$ and $\Omega$ is
\begin{equation}\label{Eq:F_Ra_1_given_R_L}
F_{\Ra_1|R_L,\Omega}(r | R_L,\Omega) = 1 - \left(\frac{R_L^2 - r^2}{R_L^2}\right)^\Omega, \ \ \ 0 \leq r \leq R_L.
\end{equation}
\end{lemma}
\begin{proof}
See Appendix~\ref{Proof:FRa1}.
\end{proof}

\noindent Let $\ncalA = \b(\origin, R_L) \backslash \b(\origin, \Ra_1)$, where $\b(\btheta, r)$ represents a ball of radius $r$ centered at $\btheta$. We have the following lemma characterizing the distribution of the $\Omega-1$ remaining active BSs in the annular region $\ncalA$.

\begin{lemma}\label{Lemma:AnnularBPP}
Conditioned on $\Ra_1$ and $R_L$, the $\Omega-1$ active BSs located inside the annular region $\b(o, R_L) \backslash \b(o, \Ra_1)$ are distributed according to a uniform BPP.
\end{lemma}
\begin{proof}
See Appendix~\ref{Proof:AnnularBPP}.
\end{proof}

%%%%%%%%%%%%%%%
% Localization Performance %
%%%%%%%%%%%%%%%

\section{Localization Performance}
As mentioned in Section~\ref{Sec:Introduction}, localization performance fundamentally depends on the number of positioning observations, their accuracies, and the locations of the participating BSs relative to the device being localized. These three factors are all highly interdependent. For instance, both the number of BSs whose signals arrive with sufficiently-high SINRs to participate in the localization procedure and the quality of their positioning observations are strongly dependent on the interference field, which is itself driven by the realized network deployment. Thus, a full characterization of localization performance would take into account all possible \emph{geometric conditionings} of the BSs over all possible channel conditions, a task which is extremely difficult. By taking into account even just the number of participating BSs averaged over the channel conditions, however, valuable insights into localization performance may be gleaned.

\subsection{Base station participation and geolocation performance}

A clear relationship exists between the number of BSs involved in positioning and a cellular network operator's ability to meet its geolocation performance requirements \cite{R1-091912}. In order to see this, consider handset-based TDOA positioning, such as Observed Time-Difference-Of-Arrival or OTDOA in LTE. The FCC E911 mandate requires that operators employing handset-based localization estimate the locations of mobile devices to within an accuracy of 50 m 67\% of the time and 150 m 90\% of the time \cite{FCCE911CFR}. Using the model described in the previous section with a BS deployment density equivalent to that of an infinite hexagonal grid with 500 m intersite distances, a shadowing standard deviation of 8~dB, and a pathloss exponent of $\alpha = 3.76$, an example localization system was simulated requiring a pre-processing SINR $\threshold/\pg=-13$~dB for successful signal detection \cite{Fischer2014}. A fully-loaded network was assumed with no BS coordination (i.e., $p=q=1$) and clock inaccuracies at each BS modeled normally with a standard deviation of 100 ns. The SINR threshold determines which BSs may successfully participate in the localization procedure, while an individual ranging observation is modeled as the sum of an exponentially-distributed non-line-of-sight bias with a mean of 30 m and a zero-mean Gaussian random variable with a variance based on the exact SINR and calculated using the well-known TOA ranging CRLB~\cite{Urkowitz1983} using a signal bandwidth of 10~MHz. Lastly, positions were estimated using the algorithm presented in \cite{Chan1994} with the strongest BS selected as the reference BS. The resulting localization performance is shown in Figure~\ref{Fig:fccE911motivation} for several minimum BS hearabilities. From this particular example, a network operator might conclude that, with universal frequency reuse on the positioning signals, a hearability of 6 BSs would be required to be in compliance with the requirements of the FCC E911 mandate.

%\begin{figure}%
%\centering%
%\begin{minipage}{0.475\textwidth}%
%\centering%
%\includegraphics[width=\figurewidth]{CRLBRange95_v_L}%
%\caption{\textsc{Relationship between localization accuracy and BS participation}: Ranges of the lower $95^\text{th}$ percentiles of positioning accuracies using TDOA positioning.}%
%\label{Fig:CRLBRange95_v_L}%
%\end{minipage}\hfill%
%\begin{minipage}{0.475\textwidth}%
%\centering%
%\includegraphics[width=\figurewidth]{CRLB_CDFs_v_MinL}%
%\caption{\textsc{Impact of minimum BS participation}: CDFs of TDOA positioning error when various minimum numbers of successful BS connections are available.}%
%\label{Fig:CRLB_CDFs_v_MinL}%
%\end{minipage}%
%\end{figure}

%\begin{figure}
%\centering
%\includegraphics[width=\figurewidth]{CRLBRange95_v_L}
%\caption{\textsc{Relationship between localization accuracy and BS participation}: Ranges of the lower $95^\text{th}$ percentiles of achievable positioning accuracies using TDOA positioning.}
%\label{Fig:CRLBRange95_v_L}
%\end{figure}
%
%\begin{figure}
%\centering
%\includegraphics[width=\figurewidth]{CRLB_CDFs_v_MinL}
%\caption{\textsc{Impact of minimum BS participation}: CDFs of TDOA positioning error when various minimum numbers of successful BS connections are available.}
%\label{Fig:CRLB_CDFs_v_MinL}
%\end{figure}

%\begin{figure}
%\centering
%\includegraphics[width=\figurewidth]{MinL_Availability_v_L}
%\caption{\textsc{Decreasing availability of BSs}: A histogram illustrating the decreasing probability of attaining an increasing number of successful BS connections.}
%\label{Fig:MinL_Availability_v_L}
%\end{figure}

\begin{figure}
\centering
\includegraphics[width=\figurewidth]{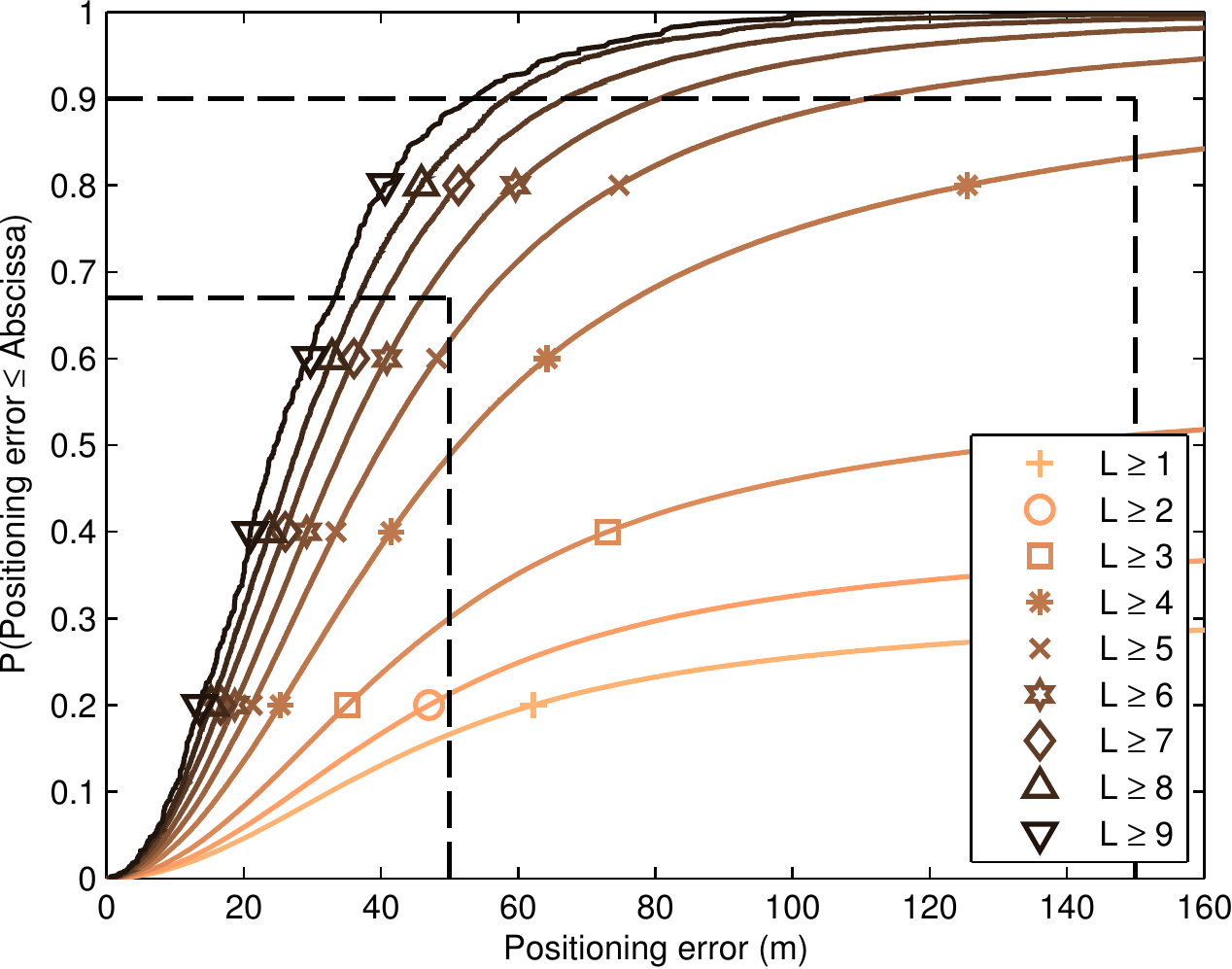}
\caption{\textsc{Impact of BS participation on meeting FCC E911 requirements}: Positioning performance of an OTDOA-like system with settings derived from \cite{Fischer2014} (see text). A performance curve must pass through \emph{both} horizontal dashed lines in order to meet the FCC mandate.}
\label{Fig:fccE911motivation}
\end{figure}

Motivated by this connection, we now formally propose a metric to study the number of BSs able to successfully participate in a localization procedure for a given target SINR $\threshold$ and processing gain $\pg$.

\begin{definition}[Participation metric]
For a given BS deployment $\pppa \in \PPPa$, let $\Upsilon$ represent the maximum number of selectable BSs such that all successfully participate in a localization procedure. This can be mathematically defined as:
\begin{equation}
\Upsilon = \argmax_\ell\ \ell \times \prod_{k=1}^\ell \indicator \left(\SINR_k(\ell) \geq \frac{\threshold}{\gamma} \right).
\end{equation}
\end{definition}

Note that this metric differs from a traditional metric such as the CRLB in that it is not directly tied to a specific positioning accuracy value, though it is still closely related. Its advantage lies in its tractability when characterized over all the possible BS topologies, whereas the CRLB, on the other hand, does provide direct insight into achievable positioning accuracy, but only for \emph{deterministic networks}, quickly becoming intractable to characterize over all BS topologies.

\subsection{Participation from a desired number of base stations}
While generally better localization performance can be attained by increasing BS participation, the probability of obtaining some desired number of successful BS connections is known to decrease sharply as the number increases. How this fundamental factor influencing localization system performance changes across different parameter sets is exactly the objective of our study. Thus, it is desirable to understand exactly how $\Upsilon$ is impacted by the network design parameters and propagation effects, averaged over the geometric conditions and channel variations.

\begin{definition}[$L$-localizability probability]
For a given $\pppa \in \PPPa$, a mobile device is said to be $L$-localizable if at least $L$ BSs may successfully participate in the localization procedure. The probability of this occurring is simply:
\begin{align}
\pl
&= \P \left( \Upsilon \geq L \right)
= \E \left[\prod_{k=1}^L \indicator \left(\SINR_k(L) \geq \frac{\threshold}{\pg} \right) \right]. \label{Eq:pl}
\end{align}
%Note that for $L=1$, this simply gives the downlink coverage probability.
\end{definition}

It may be insightful to step back and consider a specific application of $\pl$. Note that the first objective in any location system is to make sure that the device to be located can detect positioning signals from a sufficient number of BSs (i.e., that it is \emph{localizable}). By this, we simply mean that an estimate of the device's location can be found without ambiguity. In the noiseless case, this means that there can only be one solution. In the noisy case, this means that there is a single global minimum to the appropriate cost function. Commonly-accepted minimum values of $L$ for the unambiguous operation of a localization system in the $\R^2$ plane are 2, 3, and 4 for AOA, TOA/RSS, and TDOA, respectively. Thus, for example, $\mathtt{P_4}$ can be equivalently thought of as the coverage probability of TDOA positioning.

For two edge cases of interest, that with no BS coordination ($p=1$) and that with perfect BS coordination ($p=0$), it is  straightforward to infer from \eqref{Eq:SINR_k_L} that
\begin{equation}\label{Eq:SINRrelationship}
\indicator \left(\SINR_k(L) \geq \threshold \right) \geq \indicator \left(\SINR_l(L) \geq \threshold \right)
\end{equation}
for all $k \leq l \leq L$. This simply means that the received SINR from a BS farther from the mobile device is lower than that of a closer BS, implying that the probability of $L$-localizability in \eqref{Eq:pl} can be equivalently expressed as
\begin{equation}\label{Eq:pl_2}
\pl = \E\left[\indicator \left(\SINR_L(L) \geq \threshold/\pg \right) \right].
\end{equation}
With partial BS coordination ($0 < p < 1$), the relationship in \eqref{Eq:SINRrelationship} does not hold in certain corner cases. These rare occurrences have a minimal impact on our analysis, and we proceed by using the expression for $\pl$ in \eqref{Eq:pl_2} for all $p$. In the numerical results, we will validate that \eqref{Eq:pl_2} is in fact an accurate approximation of $\pl$ for all $p$ by comparing with the true $\pl$, which jointly considers the SINRs of all participating BSs.

\subsection{\texorpdfstring{A simple bound on $\pl$}{A simple bound on P\_L}}\label{Sec:pl_ubound}
We now derive an upper bound on $L$-coverage probability by (i) considering the closest interferer's location exactly, (ii) placing the remaining $\Omega-1$ active interior interferers at the edge of the field of participating BSs, and (iii) ignoring thermal noise and all distant BSs. We illustrate this setup in Figure~\ref{Fig:UpperBound_Using_Model_Example}, which applies the above considerations to the network realization of Figure~\ref{Fig:Model_Example}. Note that the interference from all BSs beyond the $L^\text{th}$ BS is ignored, while the active interior BSs have been pushed out to the same distance as the $L^\text{th}$ BS. The dominant interferer, however, has been left untouched.

\begin{figure}
\centering
\includegraphics[width=.8\columnwidth]{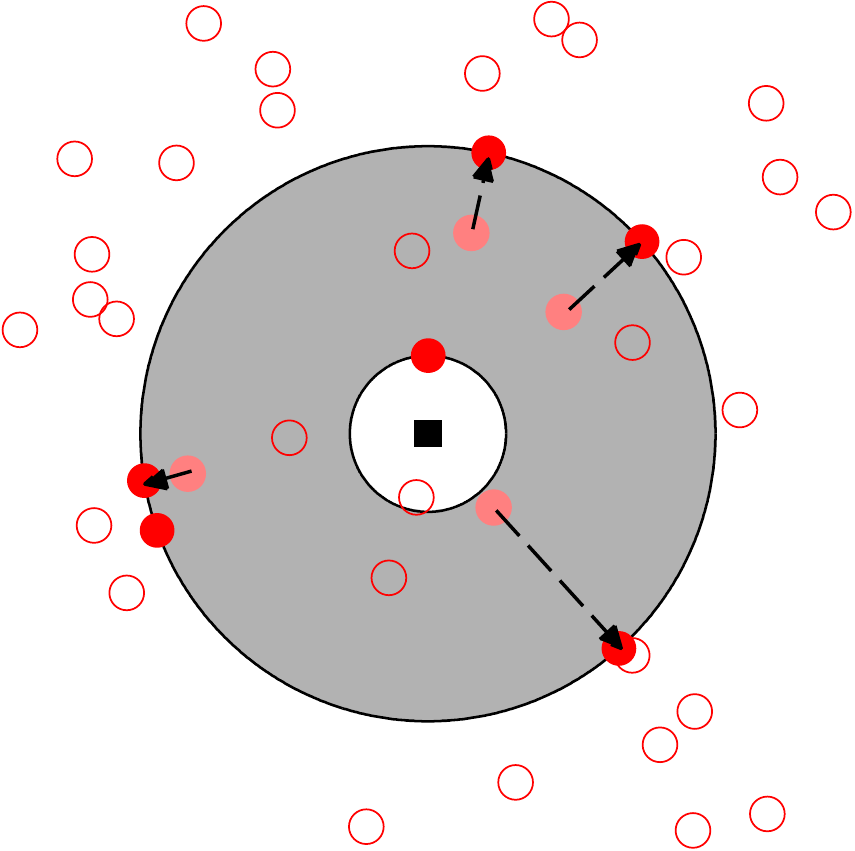}
\caption{\textsc{Lower bounding the interference}: As in Figure~\ref{Fig:Model_Example}, transmitting/quiet BSs are denoted by filled/hollow circles and the device to be localized is denoted by the filled square.}
\label{Fig:UpperBound_Using_Model_Example}
\end{figure}

\begin{theorem}[Upper bound on $L$-coverage probability]\label{Theorem:pl_ubound}
The probability that a device can successfully utilize a desired $L$ BSs for positioning is upper bounded by
\begin{align}
\pl(p, q, \alpha, \threshold, \pg, \pppai)
&\leq \sum_{\omega=0}^{\chi} \left(1-\left(\frac{\pg}{\threshold} - (\omega-1)\right)^{-\frac{2}{\alpha}} \right)^{\omega} \pomega, \label{Eq:pl_ubound}
\end{align}
where $\chi = \min\{L-1,\lfloor \pg/\beta \rfloor\}$.
\end{theorem}
\begin{proof}
See Appendix~\ref{Proof:pl_ubound}.
\end{proof}
\noindent One thing we notice immediately from \eqref{Eq:pl_ubound} is that the bound is independent from the BS deployment density $\pppai$. This is in line with similar observations made in coverage analyses for interference-limited networks (i.e., network interference dominates thermal noise)~\cite{Andrews2011,Dhillon2012}. We will return to this observation at the end of this section. In the special case when all participating BSs transmit simultaneously, the above bound reduces to a simple closed-form expression.
\begin{corollary}[The special case of $p=1$]\label{Corollary:pl_ubound_p=1}
When all participating BSs transmit simultaneously,
\begin{equation}
\pl(1, 1, \alpha, \threshold, \pg, \pppai)
\leq \left(1-\left(\frac{\pg}{\threshold} - (L-2)\right)^{-\frac{2}{\alpha}} \right)^{L-1},
\label{Eq:pl_ubound_p=1}
\end{equation}
for $\threshold < \frac{\pg}{L-1}$, and zero otherwise.
\end{corollary}
\begin{proof}
When $p=1$, $\pOmega(L-1)=1$ and the result follows from \eqref{Eq:pl_ubound}.
\end{proof}
\noindent From this result, we see that as greater numbers of participating BSs are desired, the probability of attaining those numbers of BSs with sufficiently strong connections decreases dramatically. Note that by rearranging \eqref{Eq:pl_ubound_p=1}, we also obtain a useful lower bound on the processing gain required to achieve some desired $L$-localizabity probability.
\begin{corollary}[Lower bound on processing gain with no BS coordination]\label{Corollary:lb_pg}
The processing gain required to reach a desired $L$-coverage probability $\pl$ when $p=1$ is lower-bounded by
\begin{equation}
\pg \geq \threshold \left( \left( 1-\pl^{\frac{1}{L-1}} \right)^{-\frac{\alpha}{2}} + L - 2 \right).
\end{equation}
\end{corollary}

\subsection{\texorpdfstring{Approximations of $\pl$}{Approximations of P\_L}}
Let
\begin{equation}
\Ia = \sum_{i=2}^{\Omega} P\|\xa_i\|^{-\alpha}
\end{equation}
be the aggregate interference due to the active BSs (if any) between $\xa_1$ and $x_L$. Furthermore, let
\begin{equation}
\Ix = \sum_{j=L+1}^{\infty} b_j P\|x_j\|^{-\alpha}
\end{equation}
be the aggregate interference due to the infinite field of BSs located further than the $L^\text{th}$ BS. Clearly then, when $\Omega \geq 1$,
\begin{equation}
\SINR_{L}(L) = \frac{P R_L^{-\alpha}}{P \Ra_1^{-\alpha} + \Ia + \Ix + \sigma^2}.
\label{Eq:SINR_L_L}
\end{equation}
In this section, we will approximate \eqref{Eq:SINR_L_L} and in turn $\pl$ by making the following assumption.
\begin{assumption}\label{Assumption:Approx}
We assume that if the dominant interferer (i.e., the closest active BS) is considered \emph{exactly}, the remaining interference terms in \eqref{Eq:SINR_L_L}, namely $\Ia$ and $\Ix$, may be accurately approximated by their means conditioned on $\Ra_1$, $R_L$, and $\Omega$.
\end{assumption}
\noindent This type of \emph{dominant interferer analysis} has been employed previously with desirable results~\cite{Weber2007,Heath2013} and will yield remarkably simple, yet accurate, approximations in our analysis, as will be demonstrated later. This is because considering $\Ra_1$ exactly and $\Ia$ using its mean (conditioned on $\Ra_1$), results in an accurate approximation of the total interference due to the BSs closer than $x_L$, which is typically the performance-limiting term.

In interference-limited networks, Assumption~\ref{Assumption:Approx} allows us to replace $\SINR_L(L)$ in \eqref{Eq:pl_2} by
\begin{equation}
\SIR_{L}(L) = \frac{P R_L^{-\alpha}}{P \Ra_1^{-\alpha} + \E[\Ia | \Ra_1, R_L, \Omega] + \E[\Ix|R_L]},
\label{Eq:SIR_L_L}
\end{equation}
where SIR stands for the signal-to-interference ratio. As the deployment density grows, the assumption of interference-limited networks increases in validity, and since our work is motivated by the reuse of existing infrastructures, such as cellular networks with small cell extensions which have relatively high deployment densities, we aver that it is quite reasonable. Expressions for the above conditional means are presented in the following lemmas.
\begin{lemma}\label{Lemma:EI1}
The expected value of $\Ia$ conditioned on $\Ra_1$, $R_L$, and $\Omega$ is
\begin{align}
	\E[\Ia | \Ra_1, R_L, \Omega]
	&= \frac{2P(\Omega-1)}{2-\alpha}\cdot \frac{R_L^{2-\alpha} - \Ra_1^{2-\alpha}}{R_L^2 - \Ra_1^2},
\end{align}
for $\Omega \geq 1$, and zero otherwise.
\end{lemma}
\begin{proof}
See Appendix~\ref{Proof:EI1}.
\end{proof}
\begin{lemma}\label{Lemma:EI2}
The expected value of $\Ix$ conditioned on $R_L$ is
\begin{align}
\E[\Ix|R_L]
&= \frac{2 P \pi q\pppai}{\alpha-2} R_L^{2-\alpha},
\end{align}
for $\alpha > 2$, and unbounded otherwise.
\end{lemma}
\begin{proof}
See Appendix~\ref{Proof:EI2}.
\end{proof}
\noindent By inserting the results of Lemmas~\ref{Lemma:EI1} and~\ref{Lemma:EI2} into \eqref{Eq:SIR_L_L}, we obtain increasingly specialized (and simpler) approximations of the $L$-localizability probability. In order to proceed, however, we must first account for the case when $\Omega = 0$, in which case $\Ra_1$ in \eqref{Eq:SIR_L_L} has no meaning. Thus, we first derive an approximation of $\pl$ for the special $\Omega = 0$ case, which incidentally corresponds to the perfect coordination ($p=0$) scenario.
%\begin{proposition}
%When all BSs are able to perfectly coordinate and remain quiet during the other BSs' transmissions, the probability of $L$-localizability is approximated by
%\begin{equation}
%\pl(0, q, \alpha, \threshold, \pg, \pppai) = \int_0^\infty \indicator\left(\frac{P R_L^{-\alpha}}{\frac{2 P \pi q\pppai}{\alpha-2} R_L^{2-\alpha} + \sigma^2} \geq \frac{\threshold}{\pg} \right) f_{R_L}(r) \d r.
%\end{equation}
%\end{proposition}
%\noindent A much more tractable expression is obtained in interference-limited scenarios as follows.
\begin{proposition}\label{Proposition:pl_EI1+EI2approx_p=0}
Under Assumption~\ref{Assumption:Approx}, the probability of $L$-localizability with perfect BS coordination in interference-limited networks is
\begin{equation}\label{Eq:pl_EI1+EI2approx_p=0}
\pl(0, q, \alpha, \threshold, \pg, \pppai) = 1 - \sum_{\ell=0}^{L-1} e^{-\frac{\alpha-2}{2 q \threshold/\pg}} \frac{\left(\frac{\alpha-2}{2 q \threshold/\pg}\right)^\ell}{\ell!}.
\end{equation}
\end{proposition}
\begin{proof}
See Appendix~\ref{Proof:pl_EI1+EI2approx_p=0}.
\end{proof}
\noindent As noted in the proof of above Proposition, (19) can be interpreted as the probability that \emph{at least} $L$ BSs (modeled by a PPP) lie inside $\b\left(\origin, \sqrt{\frac{\alpha-2}{2 \pi q \pppai \threshold / \pg}} \right)$. Moreover, note that the density of the BS deployment does not appear to affect $\pl$. With this special case out of the way, we now proceed, starting with our most general result.

\begin{theorem}\label{Theorem:pl_EI1+EI2approx}
%Using the means of $\Ia$ and $\Ix$, yet considering the dominant interferer exactly
Under Assumption~\ref{Assumption:Approx}, the probability that a mobile device is able to localize itself using $L$ BSs in interference-limited networks is $\pl(p, q, \alpha, \threshold, \pg, \pppai) =$
\begin{align}
&\left(1 - \sum_{\ell=0}^{L-1} e^{-\frac{\alpha-2}{2 q \threshold/\pg}} \frac{\left(\frac{\alpha-2}{2 q \threshold/\pg}\right)^\ell}{\ell!}\right) \pOmega(0) + 
\frac{4(\pppai \pi)^L}{\GammaL}~\omegasum  \notag \\
&\int_0^\infty \int_0^{r_L} \indicator \left( \frac{r_L^{-\alpha}}{\ra_1^{-\alpha} + \frac{2 (\omega-1)}{2-\alpha}\cdot \frac{r_L^{2-\alpha} - \ra_1^{2-\alpha}}{r_L^2 - \ra_1^2} + \frac{2 \pi q\pppai}{\alpha-2}r_L^{2-\alpha}} \geq \frac{\threshold}{\pg} \right) \notag \\
&\ra_1 (r_L^2 - \ra_1^2)^{\omega-1} r_L^{2(L-\omega)-1} \omega e^{-\pppai \pi r_L^2}\ \d \ra_1\ \d r_L. \label{Eq:pl_EI1+EI2approx}
\end{align}
\end{theorem}
\begin{proof}
See Appendix~\ref{Proof:pl_EI1+EI2approx}.
\end{proof}
\noindent Thus, we obtain an expression which requires only double integration in order to determine $\pl$, and though the outside integral has an infinite integration bound, the expression is not difficult to evaluate numerically thanks to the decaying exponential term in the integrand. To appreciate the value of this approximation, consider the $p=1$ scenario in which case the first term in \eqref{Eq:pl_EI1+EI2approx} disappears and the summand need only be evaluated at $\omega=L-1$. Though the authors did not have localization in mind, the exact expression for the $k$-coverage problem in the absence of fast fading in~\cite[Corr. 7]{Keeler2013} applies directly to $\pl$ for $p=q=1$. A cursory glance at the exact expression in \cite{Keeler2013} reveals its cumbersome nature, involving sums over many multi-fold integrals even in the absence of thermal noise, which is considerably more complex in comparison to \eqref{Eq:pl_EI1+EI2approx}. Moreover, by applying a slight secondary approximation and multiplying $\E[\Ix|R_L]$ in \eqref{Eq:SIR_L_L} by $\E[R_1^2]/R_1^2$, where $\E[R_1^2] = \frac{1}{\pi \pppai}$~\cite{Haenggi2005}, \eqref{Eq:pl_EI1+EI2approx} is further reduced, yielding a single-integral expression applicable to the general case.

\begin{corollary}\label{Corollary:pl_EI1+EI2approx_singleintegral}
The general expression for $\pl$ in \eqref{Eq:pl_EI1+EI2approx} may be further (and consequently, less reliably) approximated for interference-limited networks by: $\pl(p, q, \alpha, \threshold, \pg, \pppai) = $
\begin{align}
&\pl(0, q, \alpha, \threshold, \pg, \pppai) \pOmega(0) + 2\omegasum \omega \int_1^\infty \indicator \bigg( x^\alpha + \notag \\
& \frac{2(\omega-1)}{2-\alpha}\cdot \frac{x^2 - x^\alpha}{x^2-1} + \frac{2 q x^2}{\alpha-2} \leq \frac{\pg}{\threshold} \bigg) 
\frac{\left(1-x^{-2}\right)^{\omega-1}}{x^3} \d x. \label{Eq:pl_EI1+EI2approx_singleintegral}
\end{align}
\end{corollary}
\begin{proof}
See Appendix~\ref{Proof:pl_EI1+EI2approx_singleintegral}.
\end{proof}
\noindent This expression is instantly and accurately evaluated in a numerical computation environment such as MATLAB. We now briefly describe how the introduction of the additional terms into \eqref{Eq:SIR_L_L} reduces \eqref{Eq:pl_EI1+EI2approx}. The key is that after some algebraic manipulation, \eqref{Eq:SIR_L_L} may now be expressed as follows (see \eqref{Eq:pl_EI1+EI2approx_singleintegral_appendix} in the appendix): $\SIR_{L}(L) \approx $
\begin{equation}
\frac{1}{\left(\frac{R_L}{\Ra_1}\right)^{\alpha} + \frac{2 (\Omega-1)}{2-\alpha}\cdot \frac{\left(\frac{R_L}{\Ra_1}\right)^2 - \left(\frac{R_L}{\Ra_1}\right)^\alpha}{\left(\frac{R_L}{\Ra_1}\right)^2 - 1} + \frac{2 q}{\alpha-2} \left(\frac{R_L}{\Ra_1}\right)^{2}},
\label{Eq:SIR_L_L_approx}
\end{equation}
which contains $\Ra_1$ and $R_L$ only through their ratio $X = R_L/\Ra_1$. The cumulative distribution function of this ratio admits the simple form of $F_X(x) = (1-1/x^2)^\Omega$, which is readily obtained from \eqref{Eq:F_Ra_1_given_R_L}. For the special case of $\alpha = 4$, this additional approximation needn't be used, however, as a single integral expression is obtained directly from \eqref{Eq:pl_EI1+EI2approx}.

\begin{corollary}\label{Corollary:pl_EI1+EI2approx_a=4}
Under Assumption~\ref{Assumption:Approx} and for the special case of $\alpha = 4$ in interference-limited networks,
\begin{align}
&\pl(p, q, 4, \threshold, \pg, \pppai)
= \sum_{\omega=0}^\chi \pomega \int_0^{\sqrt{\frac{\pg/\threshold-\omega}{\pi q\pppai}}} \notag \\
&\left(1 - \frac{1}{\sqrt{\pg/\threshold - \pi q\pppai r^2 + \frac{(\omega-1)^2}{4}} - \frac{\omega-1}{2}}\right)^{\omega} f_{R_L}(r)\d r, \label{Eq:pl_EI1+EI2approx_a=4}
\end{align}
where $\chi = \min\{L-1,\lfloor \pg/\threshold \rfloor\}$.
\end{corollary}
\begin{proof}
See Appendix~\ref{Proof:pl_EI1+EI2approx_a=4}.
\end{proof}
\noindent For fully-loaded networks (i.e., $p=q=1$), \eqref{Eq:pl_EI1+EI2approx_a=4} simplifies further to the following expression.
\begin{corollary}\label{Corollary:pl_EI1+EI2approx_a=4_pq=1}
Under Assumption~\ref{Assumption:Approx} and for the special case of $\alpha = 4$ in fully-loaded interference-limited networks,
\begin{align}
&\pl(1, 1, 4, \threshold, \pg, \pppai) = \int_0^{\sqrt{\frac{\pg/\threshold-(L-1)}{\pi \pppai}}} \notag \\
&\left(1 - \frac{1}{\sqrt{\pg/\threshold - \pi \pppai r^2 + \frac{(L-2)^2}{4}} - \frac{L-2}{2}}\right)^{L-1} f_{R_L}(r) \d r, \label{Eq:pl_EI1+EI2approx_a=4_pq=1}
\end{align}
for $\threshold < \frac{\pg}{L-1}$ and zero otherwise.
\end{corollary}

\noindent For larger values of $L$ and higher values of $p$, it may be reasonable to ignore the $\Ix$ term in \eqref{Eq:SIR_L_L} under the assumption that the interference due to the $\Omega$ active participating BSs dominates that of the rest of the network. By doing so, expressions not requiring any integration may be found for the special cases covered in Corollaries~\ref{Corollary:pl_EI1+EI2approx_a=4} and~\ref{Corollary:pl_EI1+EI2approx_a=4_pq=1}.
\begin{corollary}\label{Corollary:pl_EI1approx_a=4}
Under Assumption~\ref{Assumption:Approx}, $\alpha=4$, $\pl$ for networks limited by the interference from the nearby BSs participating in the localization procedure is $\pl(p, q, 4, \threshold, \pg, \pppai) =$
\begin{equation}\label{Eq:pl_EI1approx_a=4}
 \sum_{\omega=0}^{\chi} \pomega \left(1 - \frac{1}{\sqrt{\pg/\threshold + \frac{(\omega-1)^2}{4}} - \frac{\omega-1}{2}}\right)^{\omega},
\end{equation}
where $\chi = \min\{L-1, \lfloor \pg/\threshold \rfloor\}$.
\end{corollary}
\begin{proof}
Consider the limit of \eqref{Eq:pl_EI1+EI2approx_a=4} as $q \to 0$. The parenthetical term becomes independent of $r$ and may be pulled out of the integral, while the remaining integral covers the probability density function $f_{R_L}(r)$ over the entirety of its domain, thus evaluating to 1.
\end{proof}
\begin{corollary}\label{Corollary:EI1approx_a=4_pq=1}
Under Assumption~\ref{Assumption:Approx} and when $\alpha=4$, $\pl$ for fully-loaded networks limited by the interference from the nearby BSs participating in the localization procedure is
\begin{equation}\label{Eq:pl_EI1approx_a=4_pq=1}
\pl(1, 1, 4, \threshold, \pg, \pppai) = \left(1 - \frac{1}{\sqrt{\pg/\threshold + \frac{(L-2)^2}{4}} - \frac{L-2}{2}}\right)^{L-1},
\end{equation}
for $\threshold < \frac{\pg}{L-1}$ and zero otherwise.
\end{corollary}
\noindent Interestingly, when interference from the participating BSs is ignored in \eqref{Eq:pl_EI1+EI2approx_p=0} and when distant network interference is ignored in \eqref{Eq:pl_EI1approx_a=4} and~\eqref{Eq:pl_EI1approx_a=4_pq=1}, the resulting expressions are independent of the BS deployment density. This is termed as \emph{scale invariance} in the literature and is known to hold for most of the metrics that are derived from SINR when the thermal noise is negligible or ignored (interference-limited scenarios), e.g., see~\cite{Andrews2011,Dhillon2012}.

%%%%%%%%%%%%
% Numerical Results %
%%%%%%%%%%%%

\section{Numerical Results and Discussion}\label{Sec:NumericalResultsAndDiscussion}
We now study the tightness of the bound and the approximations, as well as gather insights from numerical results. For consistency, all results were gathered using a BS deployment density equivalent to that of an infinite hexagonal grid with 500 m intersite distances (i.e., $\pppai = 2/(\sqrt{3} \cdot 500^2~\text{m}^2)$) and a shadowing standard deviation of 8~dB. While our analytical expressions were derived by focusing solely on the post-processing SINR from the $L^\text{th}$ BS surpassing $\threshold$, the truth data, which is compiled through simulation using an average of 1000 BSs, contains no such approximation and considers the SINRs from all $L$ participating BSs jointly with each required to surpass $\threshold$. Let us first consider $\pl$ for $L=4$ in the case of $\alpha = 4$ and fully-loaded networks ($p=q=1$), the setup considered in Figure~\ref{Fig:AllMethods_p1_q1_a4}. To interpret the meaning of this figure, it may be helpful to provide a simple example. Consider a cellular localization system employing TDOA observations and requiring a post-processing SINR of at least $-6$~dB for successful operation after experiencing a $10$~dB processing gain due to its integration time (i.e. $\threshold/\pg=-16$~dB). Figure~\ref{Fig:AllMethods_p1_q1_a4} tells us that when all BSs in the network are fully-loaded (i.e., actively transmitting), unambiguous TDOA positioning (i.e., using $L=4$ or more BSs) is possible only 50\% of the time, a result consistent with the fact that universal frequency reuse (e.g., using synchronization signals) was deemed untenable for positioning signals in LTE. Also in this figure, the truth data is compared against the closed-form bound in Corollary~\ref{Corollary:pl_ubound_p=1}, the nice closed-form approximation in Corollary~\ref{Corollary:EI1approx_a=4_pq=1} which ignores distant interference, as well as the simple single-integral approximation in Corollary~\ref{Corollary:pl_EI1+EI2approx_a=4_pq=1}. We note that the bound is tight for the \emph{high-reliability regime}, which is the range of greatest interest (localizability probabilities greater than 0.75). For low probabilities of localizability, the bound diverges by a maximum of around 3 dB. However, this low coverage range is not very desirable for system design, anyway. In addition to the bound, the closed-form approximation performs well for $\pl > 0.6$, while the single-integral approximation provides a remarkably tight approximation which is invariably indistinguishable from truth across all SINR values. The key factor in enabling these accurate approximations is that the locations of the strongest BS (i.e., the dominant interferer) and the $L^\text{th}$ BS are considered \emph{exactly}.

\begin{figure}
\centering
\includegraphics[width=\figurewidth]{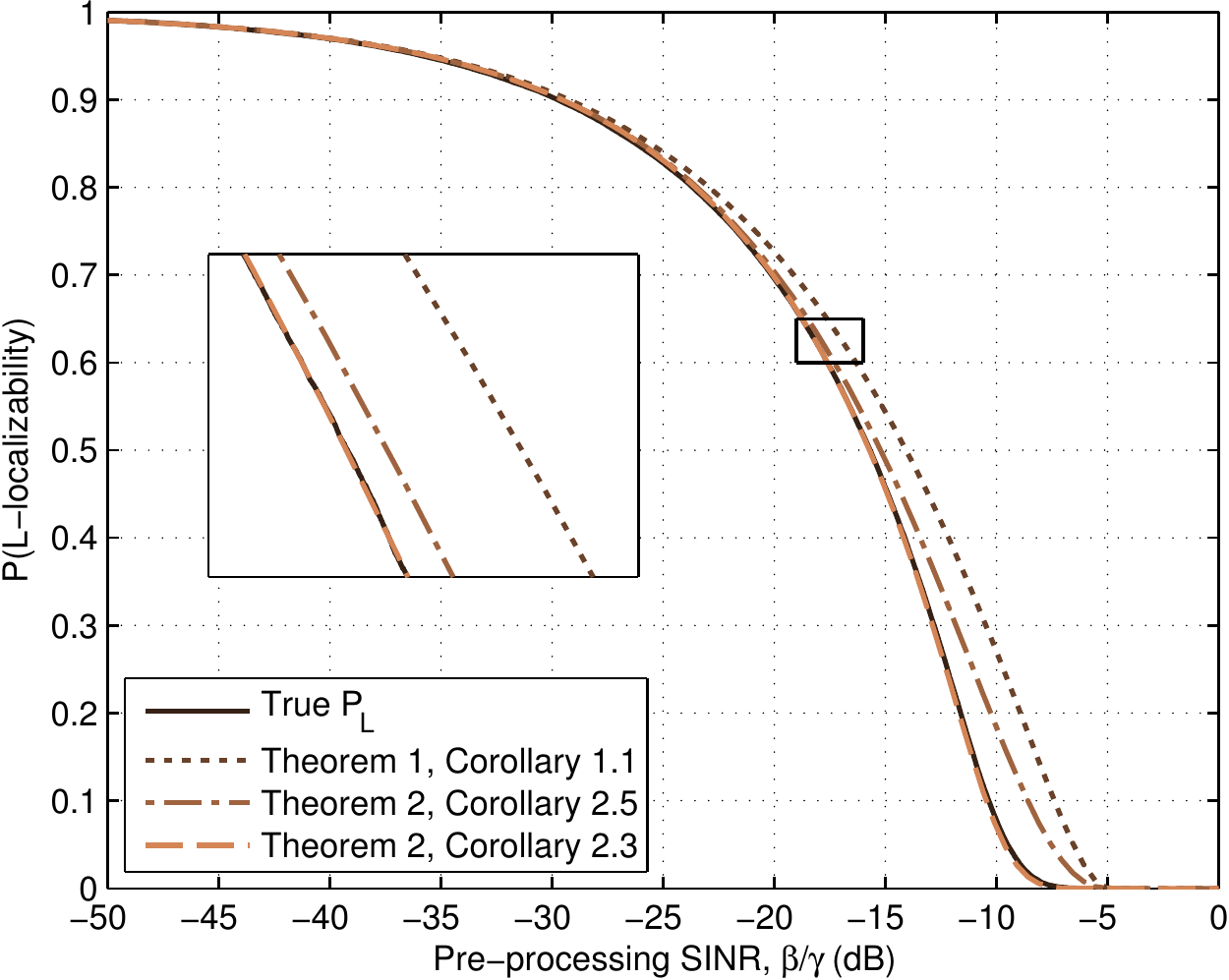}
\caption{\textsc{An initial comparison}: This figure compares the true $\pl$ with the bound and several approximations for $L=4$, $\alpha=4$, and $p=q=1$. The most accurate approximation is indistinguishable from truth, and requires a simple analytical evaluation instead of a tedious Monte Carlo approach.}
\label{Fig:AllMethods_p1_q1_a4}
\end{figure}

Inspired by the results in Figure~\ref{Fig:AllMethods_p1_q1_a4}, we now delve deeper into the approximation of Theorem~\ref{Theorem:pl_EI1+EI2approx} in \eqref{Eq:pl_EI1+EI2approx} and its single-integral counterpart in \eqref{Eq:pl_EI1+EI2approx_singleintegral} to examine how they perform under more general conditions, such as across the range of realistic pathloss exponents. Figure~\ref{Fig:Varyinga_L4_p067_q1} shows that the former remains almost indistinguishable from truth for all $\alpha$ values, while the latter is very tight in the high-reliability regime for $\alpha > 3$ and exhibits a noticeable deviation (amounting to less than 1~dB) when $\alpha=3$. For $\pl < 0.5$, the single-integral approximation deviates significantly, however, due to the fact that the interference from beyond the $L^\text{th}$ BS becomes more influential with a higher SINR requirement. This deviation is further accentuated by our intentional selection of disadvantageous $p$ and $q$ values, where $p=2/3$ partially de-emphasizes the interference from the closest BSs and $q=1$ is the worst case $q$ for the single-integral approximation since it maximally emphasizes the distant interference term where the additional approximation lies. Nevertheless, this single-integral approximation for the general parameter case performs well even under these adverse conditions for the higher $\pl$ values of interest, and will only improve with more favorable values of $p$ and $q$. On the other hand, subsequent figures will make it evident that the approximation in Theorem~\ref{Theorem:pl_EI1+EI2approx} and its exact derivatives remain extremely accurate across the ranges of all network parameter values.

\begin{figure}
\centering
\subfloat[Theorem~\ref{Theorem:pl_EI1+EI2approx}]{\label{Fig:Varyinga_L4_p067_q1_double_integral}\includegraphics[width=\figurewidth]{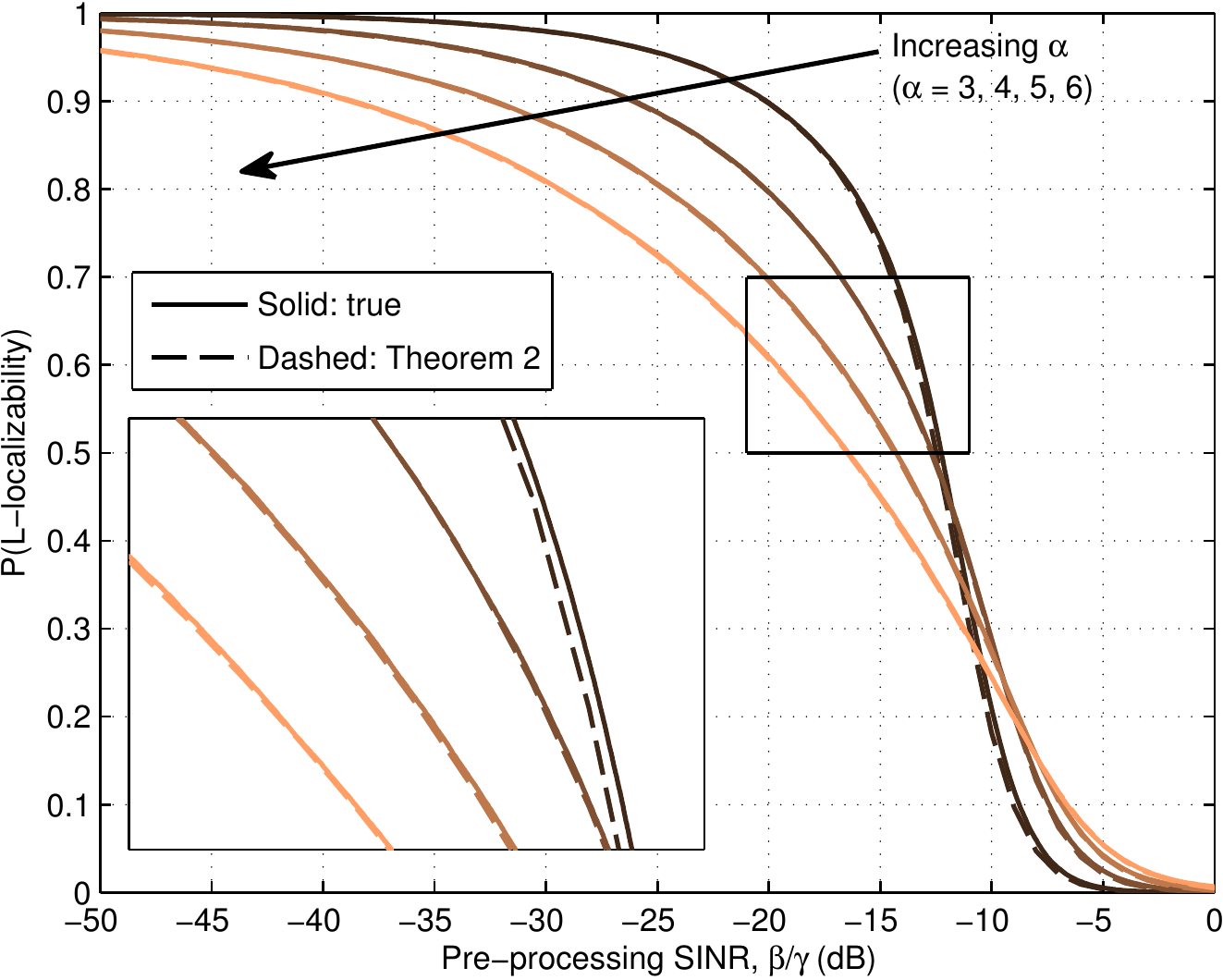}}
\ifx\sidebysidefigures\undefined
	\\
\else
	\qquad
\fi
\subfloat[Corollary~\ref{Corollary:pl_EI1+EI2approx_singleintegral}]{\label{Fig:Varyinga_L4_p067_q1_single_integral}\includegraphics[width=\figurewidth]{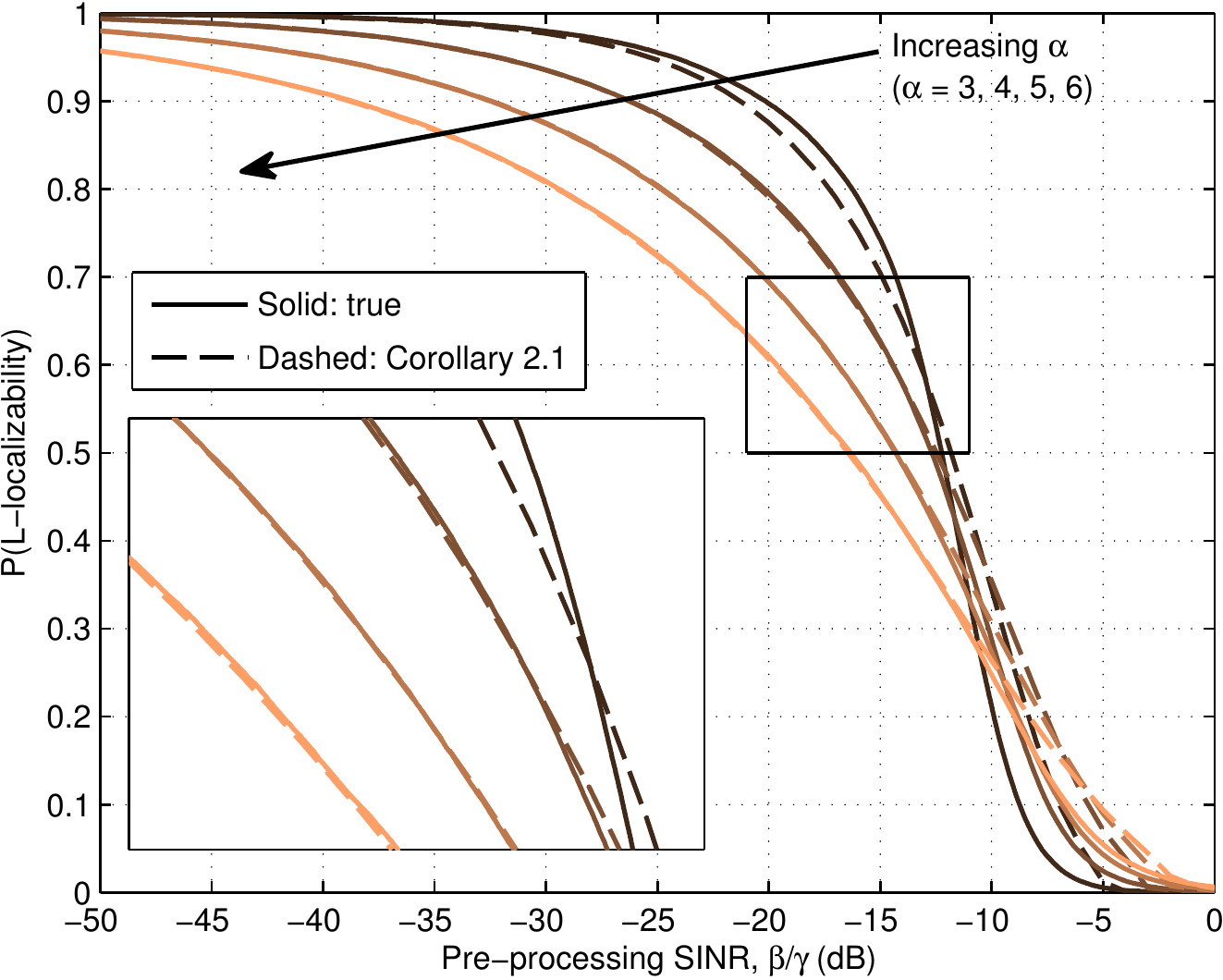}}
\caption{\textsc{General Theorem~\ref{Theorem:pl_EI1+EI2approx} approximations} across a range of pathloss exponents for $L=4$, $p=2/3$, and $q=1$.}
\label{Fig:Varyinga_L4_p067_q1}
\end{figure}

As discussed in Corollary~\ref{Corollary:lb_pg}, the upper-bound on coverage provides a useful closed-form lower-bound on the processing gain required to obtain a desired $L$-localizability probability, which is likely a specification set prior to system design. From Figure~\ref{Fig:LowerBound_ProcessingGain_pl80_a4}, it is clear that for $p=1$ and high desired values of $\pl$, the bound provides a tight criterion for the minimum processing gain required. Unfortunately, a closer look reveals that the processing gains required are impractically high, highlighting the need for some sort of interference mitigation. Moreover, consider that both bounds are very tight in the region of interest despite the fact that the interference from beyond the $L^\text{th}$ BS is fully present (i.e., $q=1$), while the bound derivation essentially assumes $q=0$. This provides more insight into the \emph{nature} of the interference, revealing that the localization performance is limited by the interference from \emph{other participating BSs} (which is consistent with conclusions drawn from LTE positioning simulations~\cite{R1-091912}), prompting one to consider coordination amongst the BSs as the desired interference mitigation technique. This becomes even more evident as the propagation conditions become more severe, such as indoors, where higher path loss exponents demand higher processing gains to achieve the same $\pl$ as seen in Figure~\ref{Fig:LowerBound_ProcessingGain_pl80_L4}.

\ifx\sidebysidefigures\undefined

\begin{figure}
\centering
\includegraphics[width=\figurewidth]{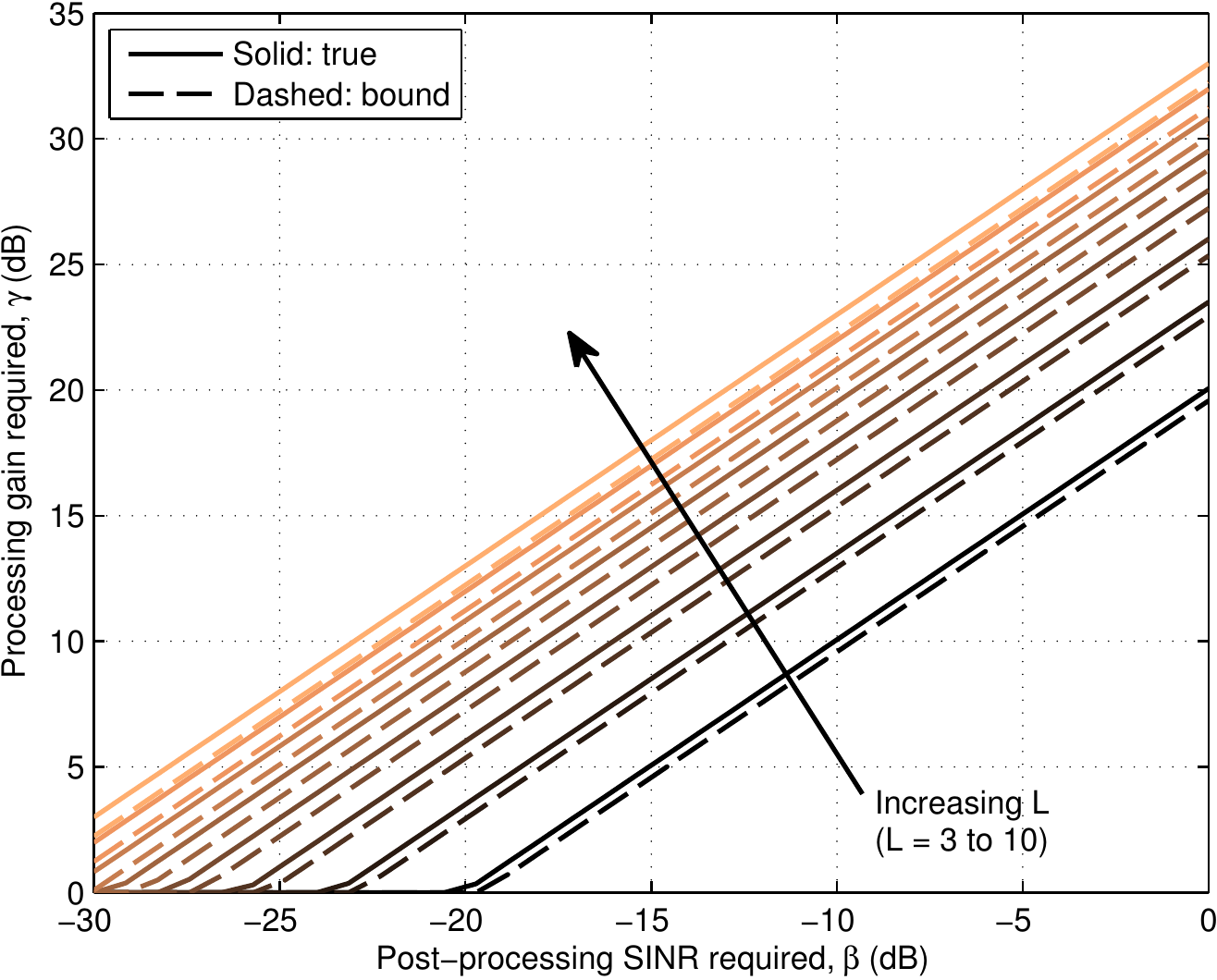}
\caption{\textsc{The importance of interference mitigation}: Processing gains required to achieve $\pl=0.8$ without any sort of interference mitigation ($p=1$) for $q=1$ and $\alpha=4$.}
\label{Fig:LowerBound_ProcessingGain_pl80_a4}
\end{figure}

\begin{figure}
\centering
\includegraphics[width=\figurewidth]{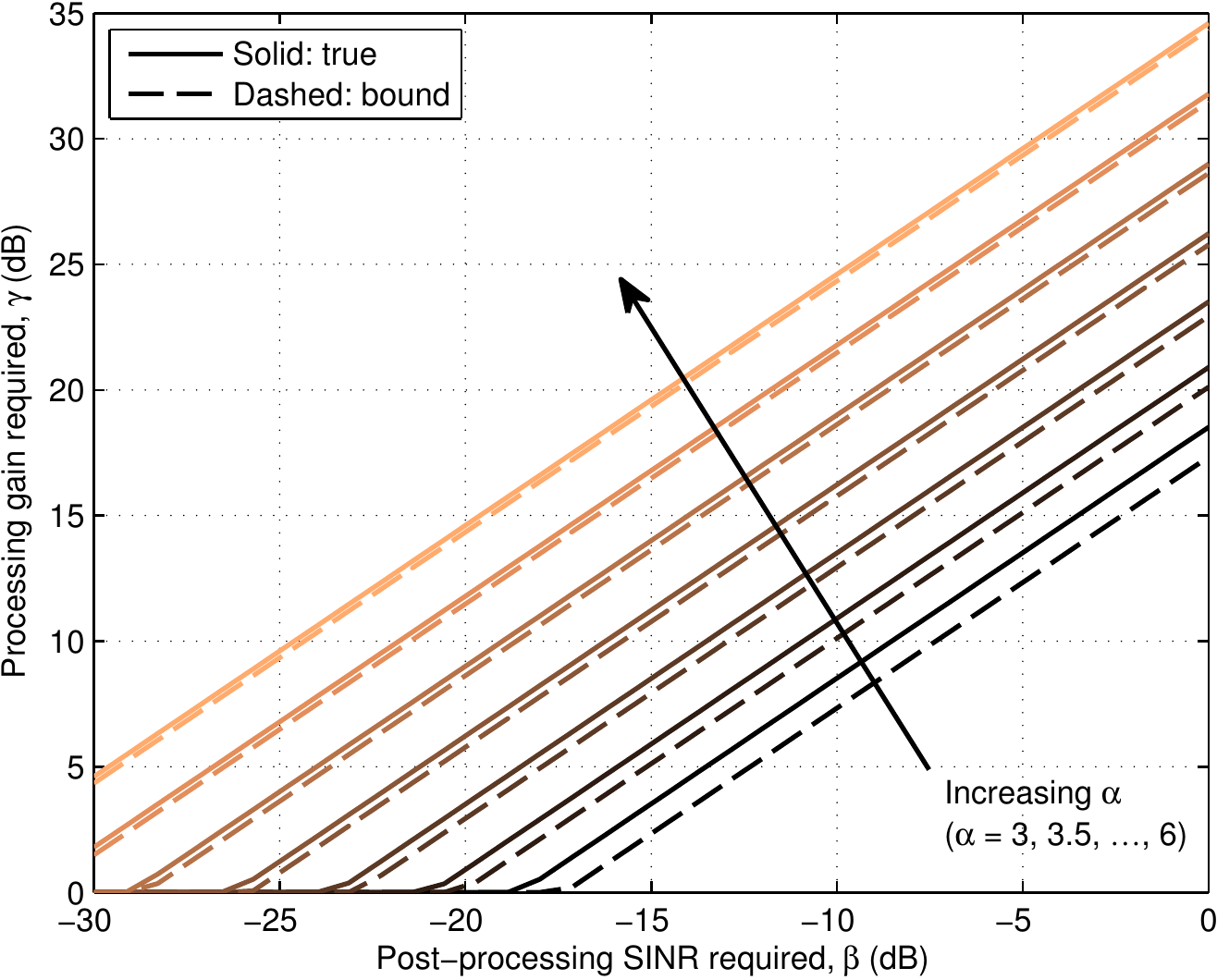}
\caption{\textsc{The impact of path loss}: Processing gains required for desired $\pl$ values grow exponentially with the pathloss exponent. ($\pl=0.8$, $L=4$, $p=q=1$.)}
\label{Fig:LowerBound_ProcessingGain_pl80_L4}
\end{figure}

\else

\begin{figure}%
\centering%
\begin{minipage}{0.475\textwidth}%
\centering%
\includegraphics[width=\figurewidth]{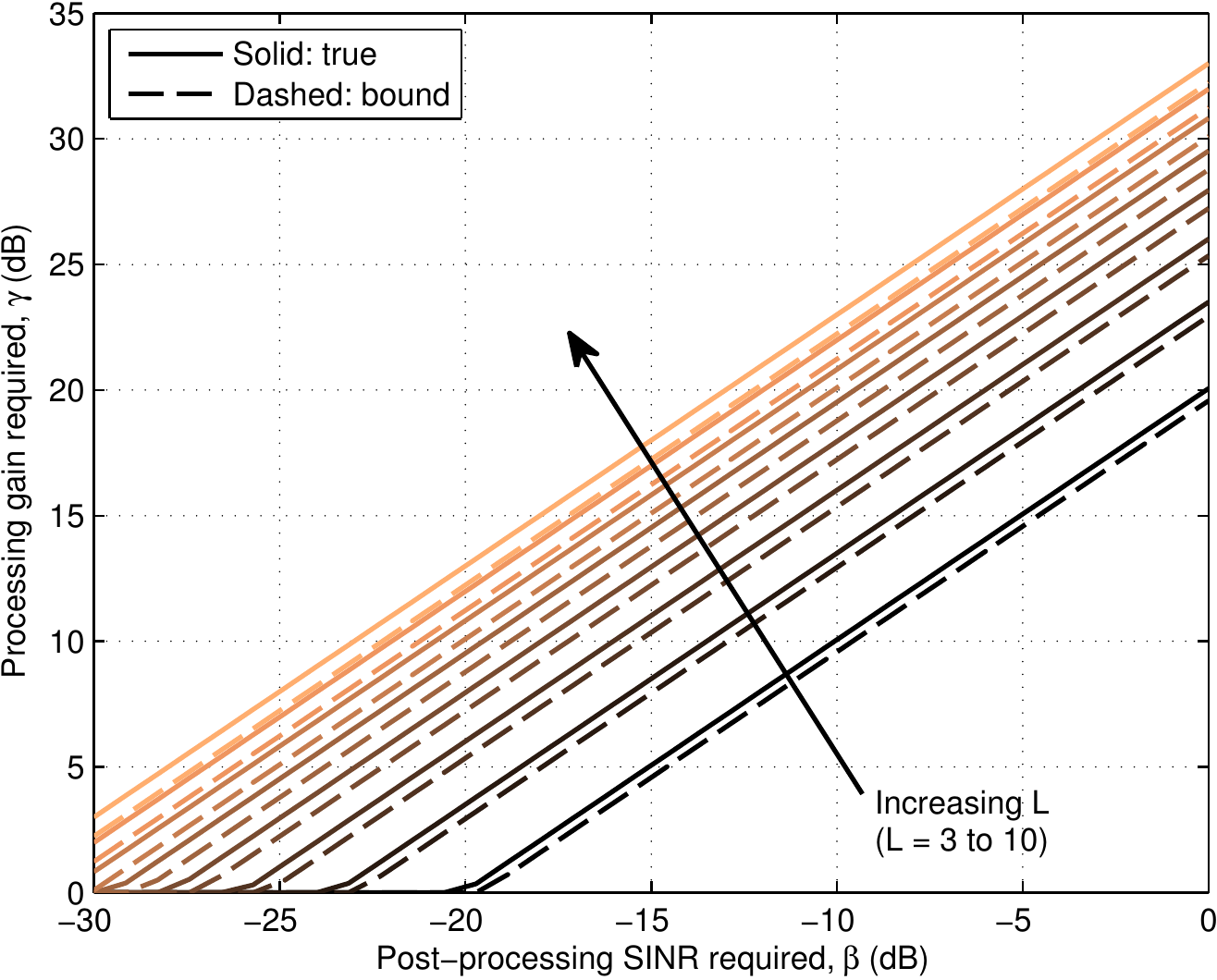}%
\caption{\textsc{The importance of interference mitigation}: Processing gains required to achieve $\pl=0.8$ without any sort of interference mitigation ($p=1$) for $q=1$ and $\alpha=4$.}%
\label{Fig:LowerBound_ProcessingGain_pl80_a4}%
\end{minipage}\hfill%
\begin{minipage}{0.475\textwidth}%
\centering%
\includegraphics[width=\figurewidth]{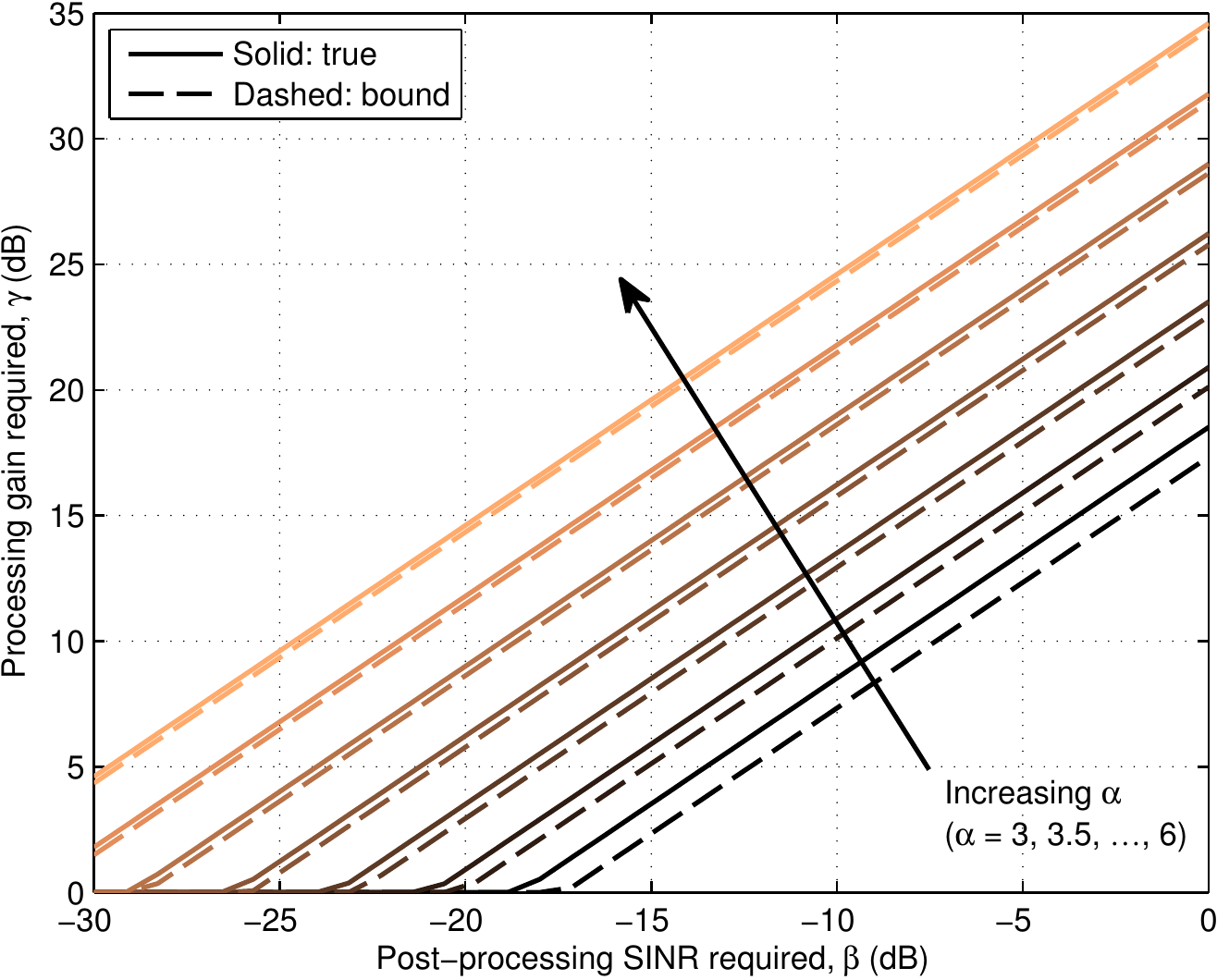}%
\caption{\textsc{The impact of path loss}: Processing gains required for desired $\pl$ values grow exponentially with the pathloss exponent. ($\pl=0.8$, $L=4$, $p=q=1$.)}%
\label{Fig:LowerBound_ProcessingGain_pl80_L4}%
\end{minipage}%
\end{figure}

\fi

Through Figure~\ref{Fig:Varyingp_L4_q1_a4}, we demonstrate the value of mitigating interference through BS coordination. With the BSs beyond the $L^\text{th}$ BS fully-loaded, we note an increasing benefit from BS coordination as the degree of the coordination is incrementally increased. The large improvement in hearability from $p=1$ to $p=0$ again sheds light on the issue of interference from nearby BSs. In fact, at its least detrimental point, the interference from the $L-1$ closest BSs accounts for nearly a 10 dB drop in signal quality. That's roughly the equivalent of going from successfully obtaining eleven to obtaining only four BSs in the network considered in Figure~\ref{Fig:VaryingL_p05_q075_a4}. While coordination is able to help significantly, a closer look at Figures~\ref{Fig:Varyingp_L4_q1_a4} and~\ref{Fig:VaryingL_p05_q075_a4} makes it clear that more is necessary in order to reach the truly high $L$-localizability probabilities at reasonable SINR values. The most prevalent technique to further reduce the interference is that of thinning the interference field by separating nearby transmitters in frequency, which is discussed next.

\ifx\sidebysidefigures\undefined

\begin{figure}
\centering
\includegraphics[width=\figurewidth]{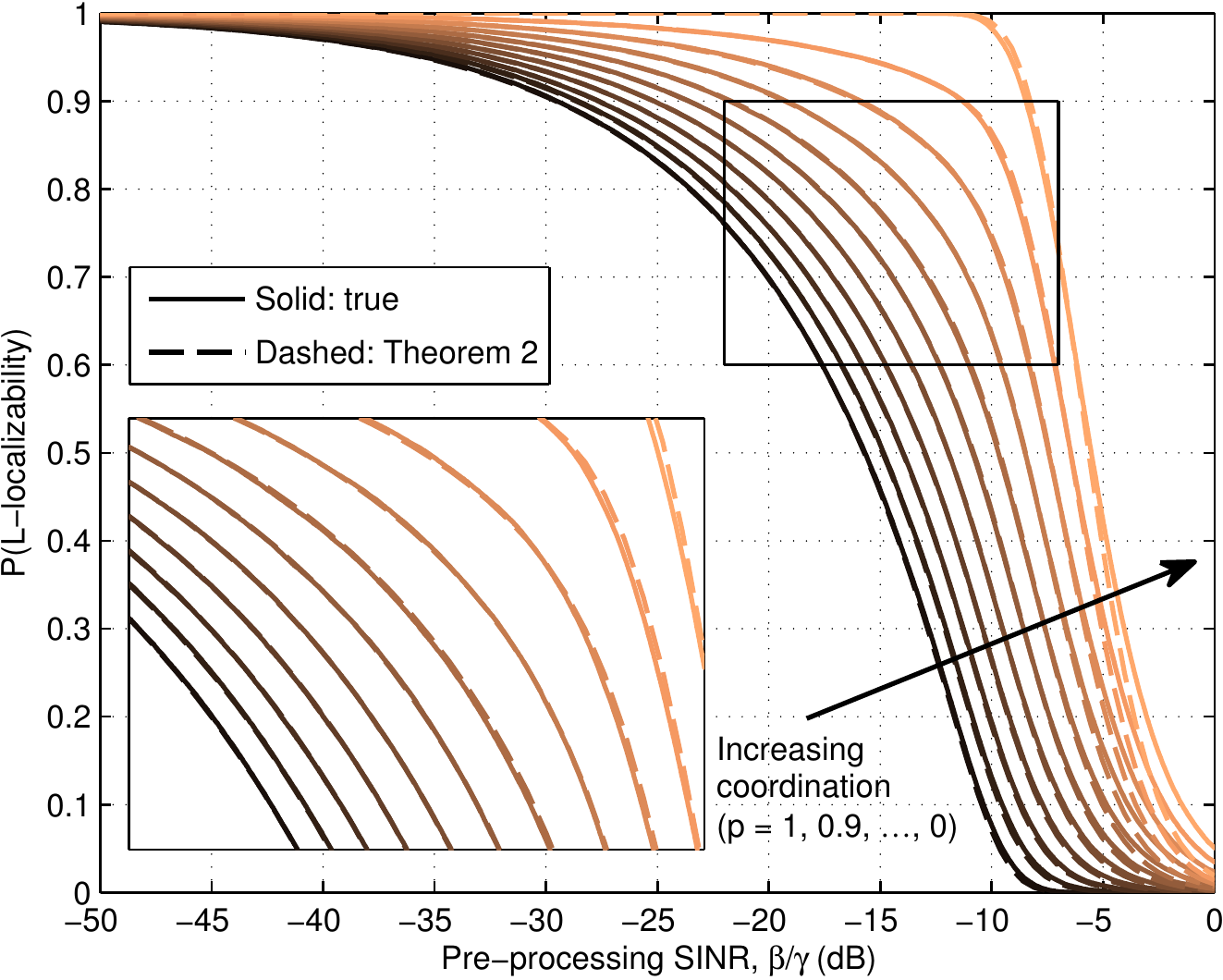}
\caption{\textsc{The benefit of BS coordination}: These $L$-localizability probabilities emphasize the gains achievable through the coordinated mitigation of participating BS interference. ($L=4$, $\alpha=4$, $q=1$.)}
\label{Fig:Varyingp_L4_q1_a4}
\end{figure}

\begin{figure}
\centering
\includegraphics[width=\figurewidth]{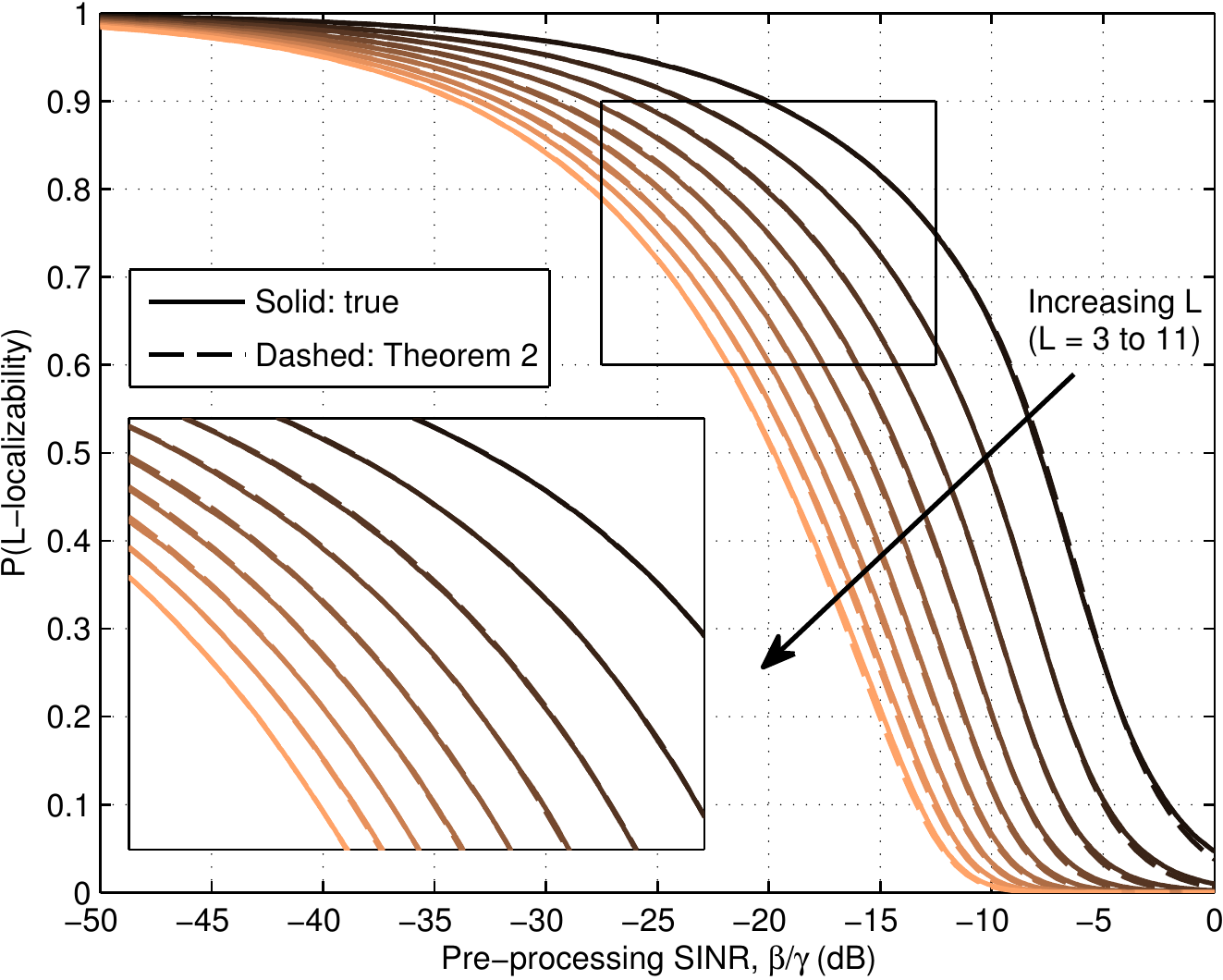}
\caption{\textsc{The cost of involving more BSs}: This figure underscores the high cost of obtaining a greater desired number of BSs, particularly in the high-reliability regime. ($\alpha=4$, $p=1/2$, $q=3/4$.)}
\label{Fig:VaryingL_p05_q075_a4}
\end{figure}

\else

\begin{figure}%
\centering%
\begin{minipage}{0.475\textwidth}%
\centering%
\includegraphics[width=\figurewidth]{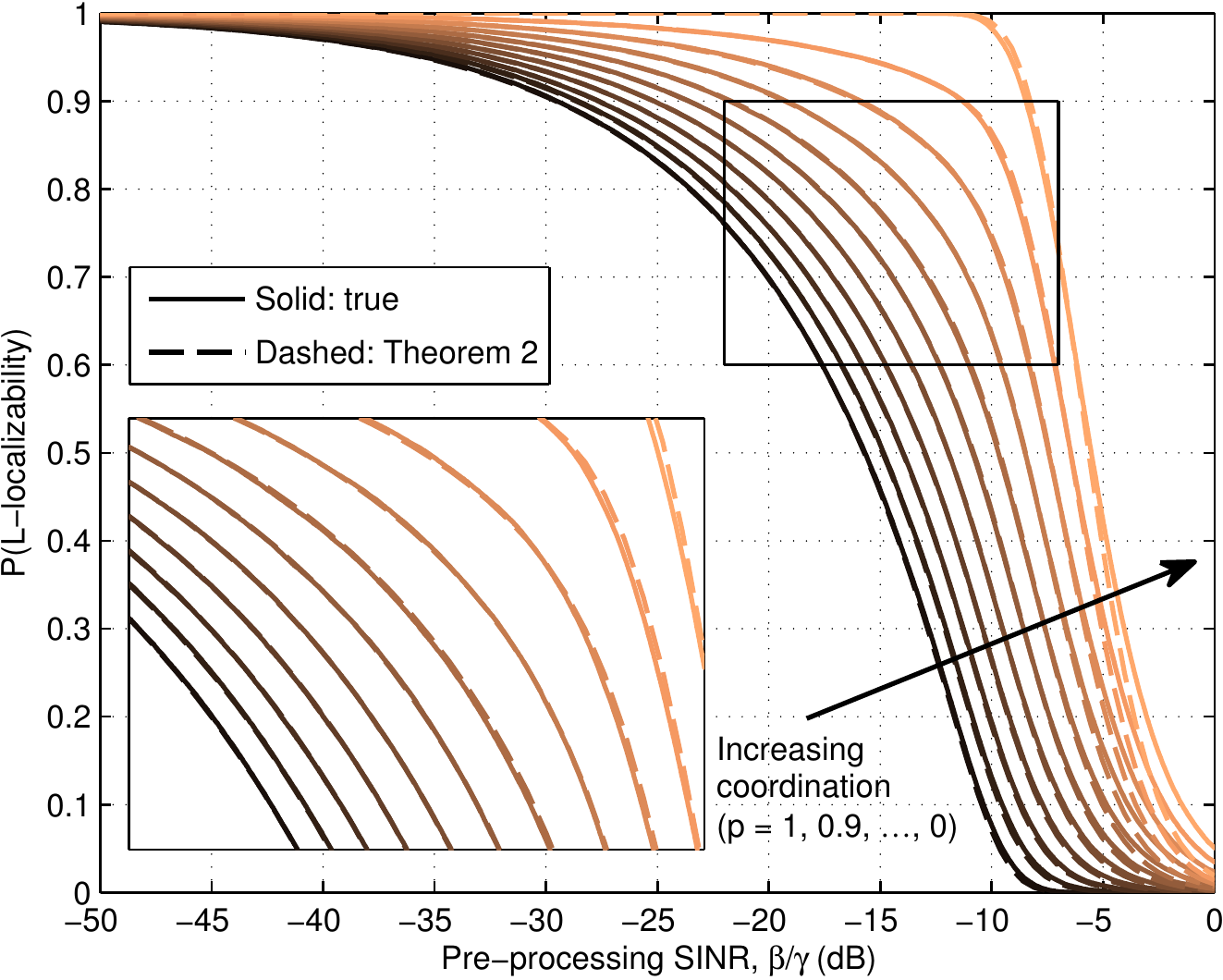}%
\caption{\textsc{The benefit of BS coordination}: These $L$-localizability probabilities emphasize the gains achievable through the coordinated mitigation of participating BS interference. ($L=4$, $\alpha=4$, $q=1$.)}%
\label{Fig:Varyingp_L4_q1_a4}%
\end{minipage}\hfill%
\begin{minipage}{0.475\textwidth}%
\centering%
\includegraphics[width=\figurewidth]{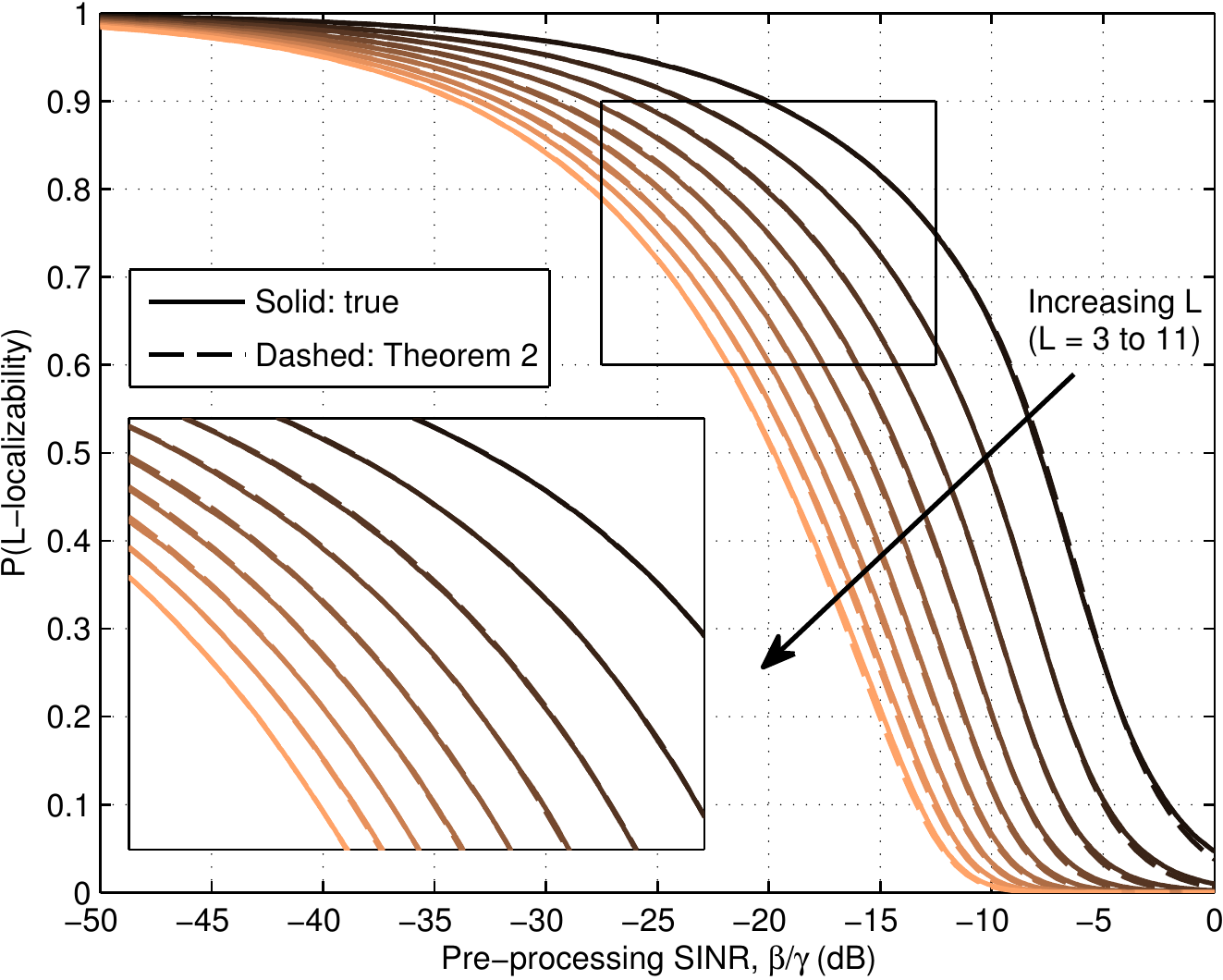}%
\caption{\textsc{The cost of involving more BSs}: This figure underscores the high cost of obtaining a greater desired number of BSs, particularly in the high-reliability regime. ($\alpha=4$, $p=1/2$, $q=3/4$.)}%
\label{Fig:VaryingL_p05_q075_a4}%
\end{minipage}%
\end{figure}

\fi

\subsection{Application to frequency reuse}
An effective technique to reduce interference from co-located transmitters is that of \emph{frequency reuse}, commonly included in wireless standards~\cite{3GPP.36.305}. If a total of $K$ frequency bands are available and we independently assign one of the bands to each $x \in \PPPa$ with equal probability, we can easily incorporate frequency reuse into our model by considering the transmission activity on each band separately using independent PPPs whose densities are that of the original PPP thinned by the frequency reuse factor $K$. Let us denote the number of detectable BSs in the $k\th$ band by $n_k$. For $L$-localizability, we thus require $\sum_{k=1}^K n_k \geq L$. The main technical difference between this analysis and that already presented is that we are now obliged to evaluate the probability of being able to detect \emph{exactly} $n$ BSs in a given frequency band. Under the localization setup we consider, it is not difficult to show that when $p=q$ this can be readily determined by considering the difference between the probabilities of $n$ and $n+1$ localizability. The $p=q$ condition isn't particularly restrictive, as it still allows us to model average network load and only removes the ability to perform per-band BS coordination (which isn't done in practice, anyway). From here, the expression for the $L$-localizability probability with frequency reuse is procured by simply considering all combinations of $\{n_1, \ldots, n_K\}$ for which $\sum_{k=1}^K n_k \geq L$. This is formally stated in the following theorem.

\begin{theorem}[$L$-localizability probability with random frequency reuse]\label{Theorem:FrequencyReuse}
The $L$-localizability probability with random frequency reuse and reuse factor $K$ is:
\begin{align}
&\plk(p, q, \alpha, \threshold, \pg, \pppai)
= 1 - \sum_{l=0}^{L-1} \sum\limits_{\substack{\{n_1, \ldots, n_K\} \\ \sum_i n_i = l}} \prod_{i=1}^K  \notag \\
&\left( {\tt P}_{n_i}\left (p, q, \alpha, \threshold, \pg, \frac{\pppai}{K}\right ) - {\tt P}_{n_i+1}\left (p, q, \alpha, \threshold, \pg, \frac{\pppai}{K}\right ) \right). \label{Eq:FrequencyReuse}
\end{align}
\end{theorem}
\noindent This expression is easily implemented in software using a recursive function and any of the previous expressions for evaluating $\pl$ may be substituted into \eqref{Eq:FrequencyReuse}. Using the worst-case interference scenario ($p=q=1$), the value of frequency reuse is clearly demonstrated in Figure~\ref{Fig:FrequencyReuse_L4_p1_q1_a4} for $K = 3$ and $6$, which correspond to the values used in LTE when positioning using cell-specific reference signals (CRS) and positioning reference signals (PRS), respectively. Since random frequency reuse makes no attempt to intelligently separate interference sources, gains from planned frequency reuse are likely to be even better.

\begin{figure}
\centering
\includegraphics[width=\figurewidth]{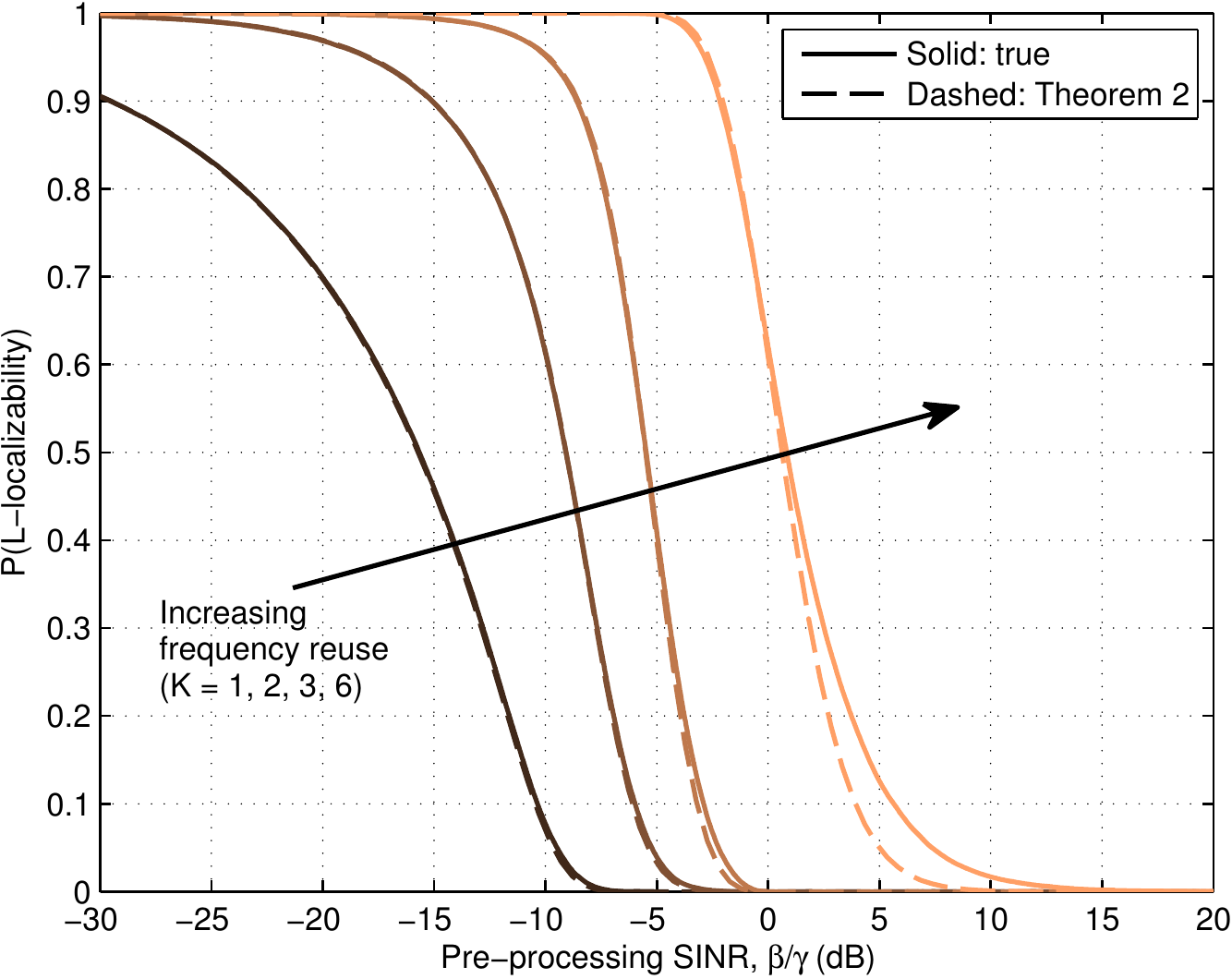}
\caption{\textsc{The value of frequency reuse} is undeniable for improving $\pl$. Here, $L=4$, $\alpha=4$, $p=q=1$.}
\label{Fig:FrequencyReuse_L4_p1_q1_a4}
\end{figure}

\subsection{Example LTE use case}
Although the model presented in this paper is general and not tied to any specific localization system, here we provide an example of how our model may be initially related to an existing technique. Consider the scenario of OTDOA positioning in LTE using the PRS, for which it is reasonable to assume that up to six close BSs may be selected to localize a mobile device without interfering with each other due to a frequency reuse factor of six among the PRS. BSs farther away may still interfere, however. This setup may then be treated as one with perfect coordination ($p=0$) among the $L=6$ participating BSs, while the expected interference from beyond the sixth BS is adjusted using $q \leq 1$ in accordance with the average network load. The processing gain $\pg \approx 15$ dB from the length-31 Gold sequences used in the PRS while $\alpha=3.76$ per 3GPP assumptions \cite{R1-091443}. Furthermore, even PRS power boosting may be accounted for in a preliminary way by considering the \emph{power boost} at the participating BSs as an equivalent \emph{power reduction} at the interfering BSs. Practically, this may be implemented by reducing the value of $q$, the ultimate value of which will depend on the amount the participating BSs' powers are boosted.

\subsection{Similarity between PPP and hexagonal grid hearability}

\ifx\sidebysidefigures\undefined

\begin{figure}
\centering
\includegraphics[width=\figurewidth]{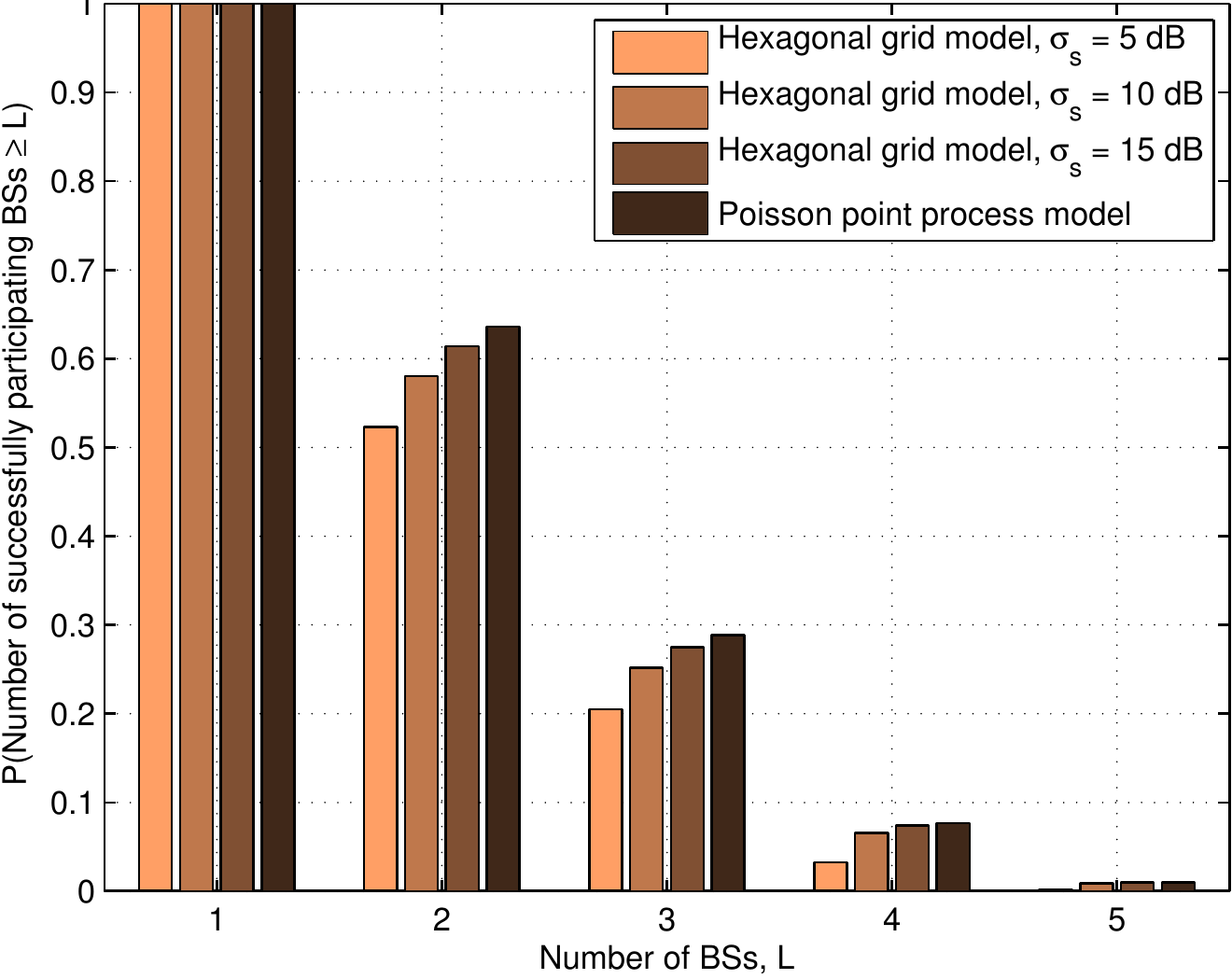}%
\caption{\textsc{Convergence to PPP}: Hearability results gathered using a regular grid model converge to those gathered using a PPP model as the shadowing standard deviation $\sigma_s$ increases. ($\threshold/\pg = -10$~dB, $\alpha=4$, $p=q=1$.)}%
\label{Fig:PPPvHexGridHearability2}%
\end{figure}

\begin{figure}
\centering
\includegraphics[width=\figurewidth]{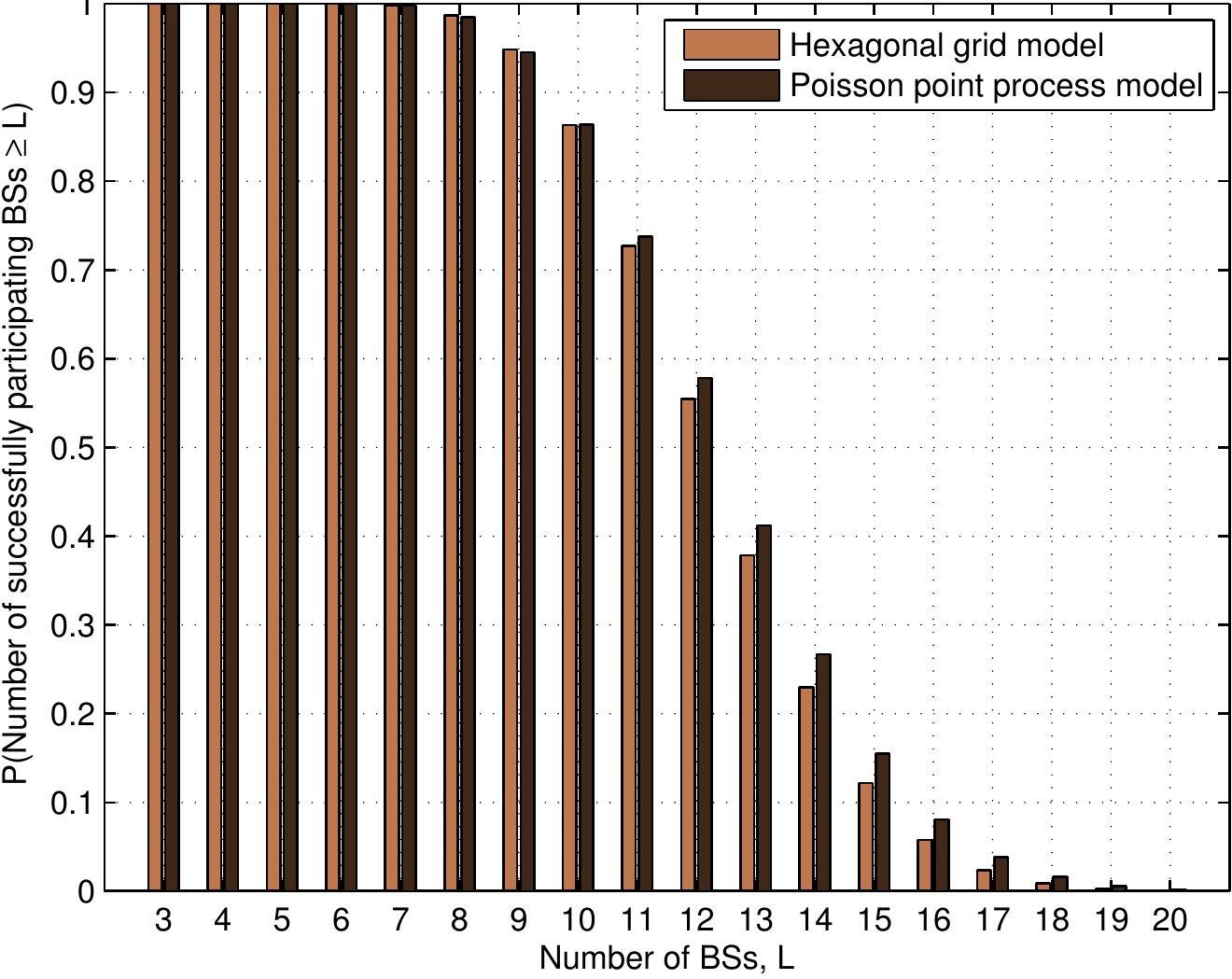}%
\caption{\textsc{Comparison with Frequency Reuse}: When $K>1$, the models agree for greater values of $L$ even at a lower $\sigma_s = 8$~dB. A required $L > 10$ is not typically expected. ($\threshold/\pg = -10$~dB, $\alpha=4$, $p=q=1$, $K=6$.)}%
\label{Fig:PPPvHexGridHearability1}%
\end{figure}

\else

\begin{figure}%
\centering%
\begin{minipage}{0.475\textwidth}%
\centering%
\includegraphics[width=\figurewidth]{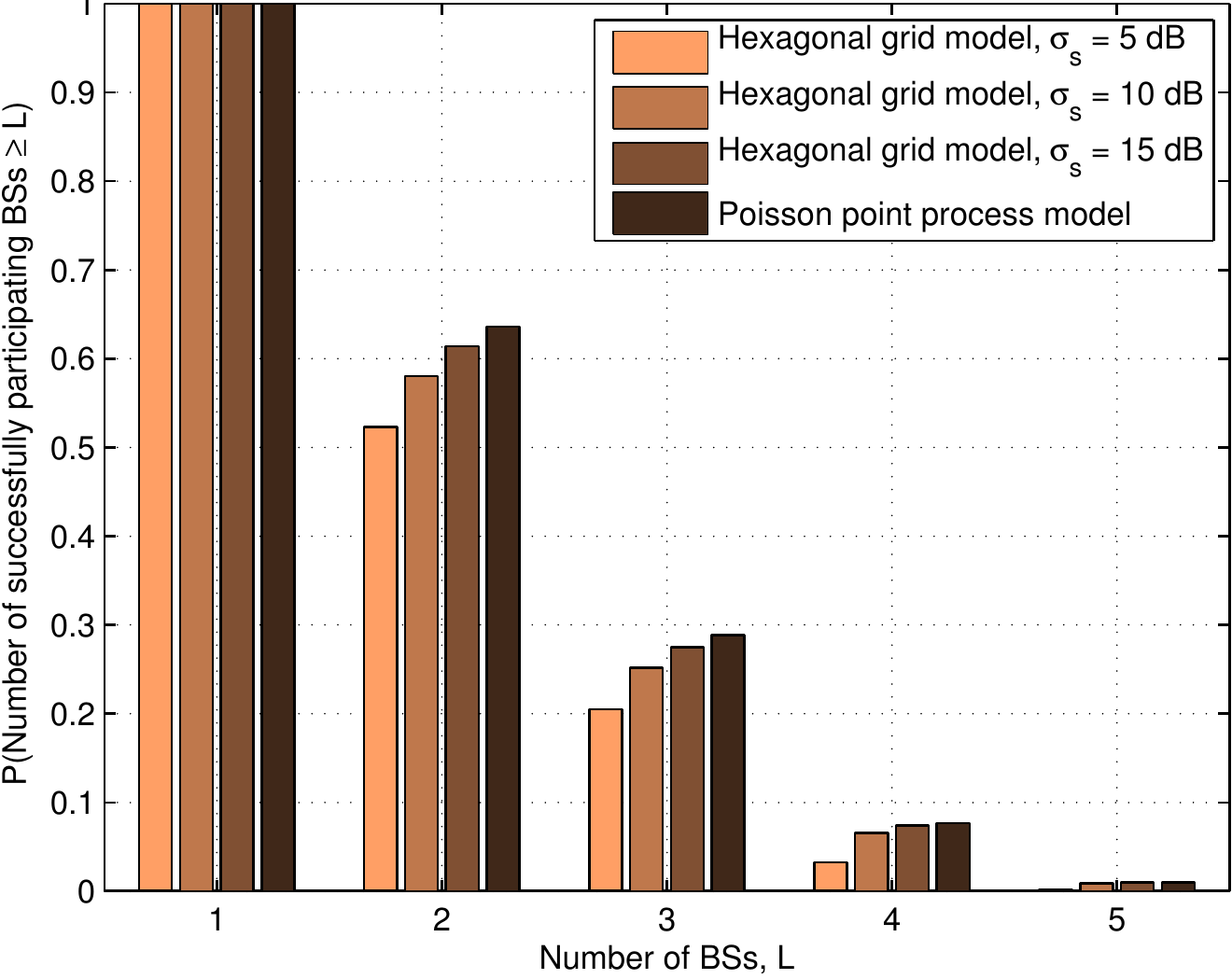}%
\caption{\textsc{Convergence to PPP}: Hearability results gathered using a regular grid model converge to those gathered using a PPP model as the shadowing standard deviation $\sigma_s$ increases. ($\threshold/\pg = -10$~dB, $\alpha=4$, $p=q=1$.)}%
\label{Fig:PPPvHexGridHearability2}%
\end{minipage}\hfill%
\begin{minipage}{0.475\textwidth}%
\centering%
\includegraphics[width=\figurewidth]{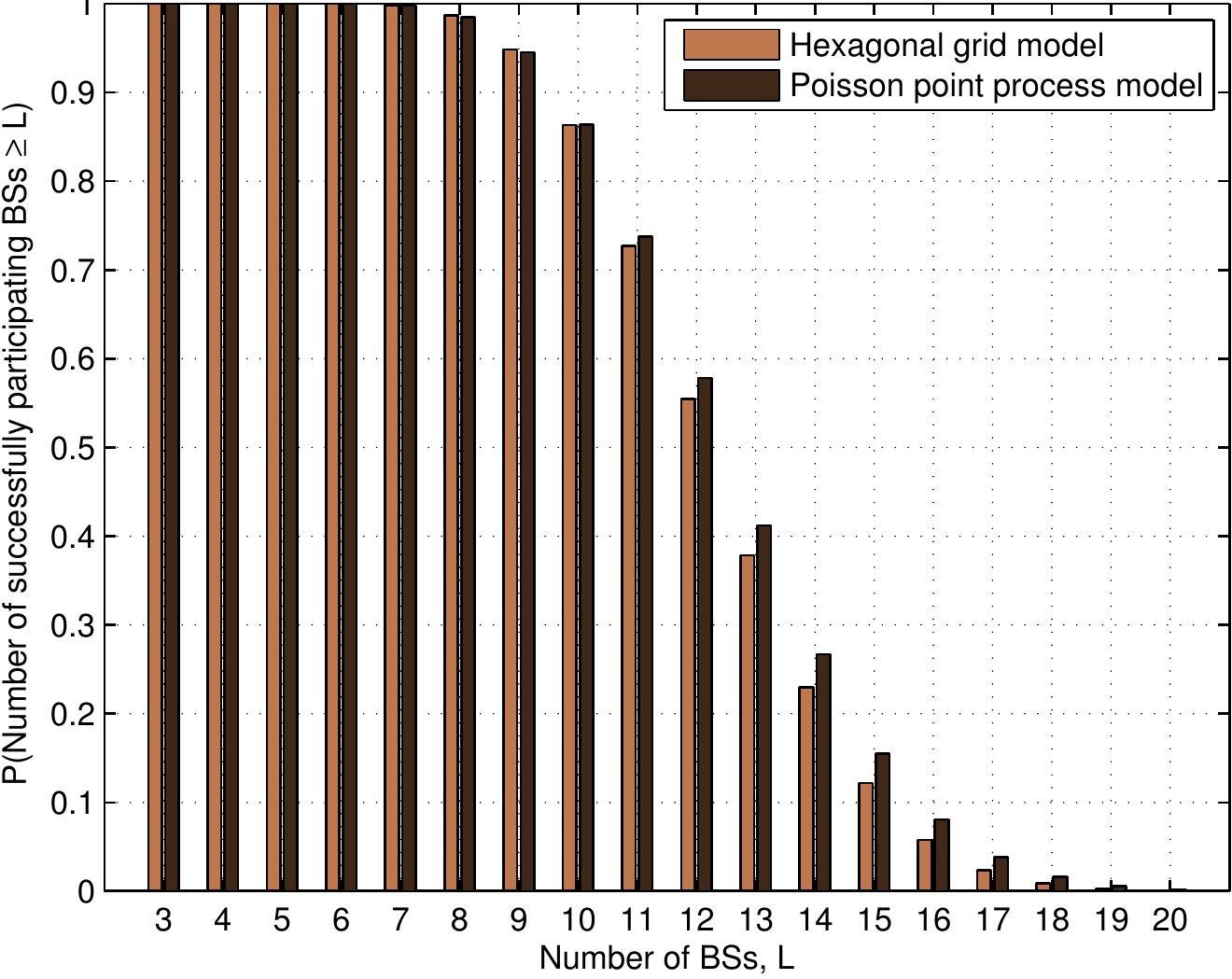}%
\caption{\textsc{Comparison with Frequency Reuse}: When $K>1$, the models agree for greater values of $L$ even at a lower $\sigma_s = 8$~dB. A required $L > 10$ is not typically expected. ($\threshold/\pg = -10$~dB, $\alpha=4$, $p=q=1$, $K=6$.)}%
\label{Fig:PPPvHexGridHearability1}%
\end{minipage}%
\end{figure}

\fi

We now briefly revisit our assumption of BSs distributed according to a PPP and show that this model yields hearability results which are comparable to those yielded by the well-accepted hexagonal grid model. As discussed in \cite{Blaszczyszyn2013,Keeler2014,Blaszczyszyn2015}, increased shadowing causes more regular network models to appear Poisson when considering received signal strength-based metrics. In Figure~\ref{Fig:PPPvHexGridHearability2}, hearability results using the two models are compared for several increasing shadowing strengths (note that, as discussed previously, the PPP results are invariant to shadowing strength), $p=q=1$, and $\threshold/\pg = -10~\text{dB}$, and it is clear that the models are converging. Furthermore, frequency reuse causes the models to provide similar results for higher values of $L$ even at lower shadowing strengths. This is shown in Figure~\ref{Fig:PPPvHexGridHearability1} for $K=6$, $\sigma_s = 8~\text{dB}$, and the other parameters the same as above. While there will always be some difference in the hearability results using the two models, the results using actual deployments likely lie somewhere in between the two, as was discovered in \cite{Andrews2011} for cellular coverage.

%%%%%%%%
% Conclusion %
%%%%%%%%

\section{Conclusion}
We have employed concepts from point process theory and stochastic geometry in order to provide a new tractable model for studying localization in cellular networks. The model is accompanied by an analysis of hearability, which is shown to be an important metric providing insight into fundamental localization performance, and easy-to-use analytical expressions for this metric are provided. This is in contrast to most previous approaches, which either provide insights specific to deterministic deployments or rely on time-consuming simulations, neither of which allow for the drawing of general insights. While more base stations participating in localization is beneficial to its accuracy, the analytical results show that obtaining an increasing number of successful base stations connections decreases exponentially. The primary culprit for this is the interference due to the participating base stations themselves. When base stations coordinate (e.g., through a network controller), more of them are able to successfully participate in the localization procedure. However, in the end, it is clear that the greatest gains in hearability are realized through frequency reuse. This is consistent with the findings during the standardization of LTE positioning, where this conclusion was reached after performing many complex system-level simulations.

We have only scratched the surface of the model, which may be used to extend this work in many different ways. Besides hearability, understanding how other metrics, such as ones focused on network geometry, are affected by changes in design parameters would yield valuable contributions to the localization literature. In addition, recognizing how directional antennas and power control affect localization performance would be helpful to cellular system designers. Lastly, location information can be gathered in diverse ways, for example by simply knowing that one is able to hear a low-power femtocell, which lends itself nicely to the consideration of different classes of base stations through multi-tiered heterogeneous networks~\cite{Dhillon2012}.

%%%%%%%%
% Appendices %
%%%%%%%%

\appendix

\subsection{Proof of Lemmas~\ref{Lemma:CircularBPP} and~\ref{Lemma:AnnularBPP}}\label{Proof:AnnularBPP}
The proof strategy used here is quite standard, e.g., see~\cite[Theorem 2.9]{Haenggi2013}. For the proof, let $i < j$, $n \leq \aleph$, $\aleph$ be the number of BSs between $x_i$ and $x_j$, $A = \b(\origin,\|x_i\|+\delta/2) \backslash \b(\origin,\|x_i\|-\delta/2)$, $B = \b(\origin,\|x_j\|+\varepsilon/2) \backslash \b(\origin,\|x_j\|-\varepsilon/2)$, $C = \b(\origin,\|x_j\|-\varepsilon/2) \backslash \b(\origin,\|x_i\|+\delta/2)$, $D \subseteq C$, $\b(\btheta, r)$ represent a ball of radius $r$ centered at $\btheta$, $\delta$ and $\varepsilon$ be small modifications to the regions of $A$, $B$, and $C$, and $N_\ncalH$ be the number of points in region $\ncalH$. Then, $\P(N_D=n| x_i, x_j) =$
\begin{align*}
&\lim_{\delta,\varepsilon \to 0} \P(N_D=n | N_C=\aleph, N_B=1, N_A=1) \\
&= \lim_{\delta,\varepsilon \to 0} \frac{\P(N_D=n, N_C=\aleph, N_B=1, N_A=1)}{\P(N_C=\aleph, N_B=1, N_A=1)}  \\
&= \lim_{\delta,\varepsilon \to 0} \frac{\P(N_D=n, N_{C\backslash D}=\aleph-n, N_B=1, N_A=1)}{\P(N_C=\aleph, N_B=1, N_A=1)} \\
%&= \lim_{\delta,\varepsilon \to 0} \frac{\P(N_D=n)\,\P(N_{C\backslash D}=\aleph-n) \,\P(N_B=1)\,\P(N_A=1)}{\P(N_C=\aleph)\,\P(N_B=1)\,\P(N_A=1)} \\
&= \frac{\P(N_D=n)\,\P(N_{C\backslash D}=\aleph-n)}{\P(N_C=\aleph)} \\
&= \frac{(\pppai|D|)^n e^{-\pppai|D|}/n!}{(\pppai|C|)^{\aleph} e^{-\pppai|C|}/\aleph!} \cdot \frac{(\pppai(|C|-|D|))^{\aleph-n} e^{-\pppai(|C|-|D|)}}{(\aleph-n)!} \\
&= {\aleph \choose n} %\frac{\aleph!}{n!(\aleph-n)!} 
\left (\frac{|D|}{|C|}\right )^n\left (1-\frac{|D|}{|C|}\right )^{\aleph-n},
\end{align*}
which shows that the $\aleph$ BSs inside the annular region $\lim_{\delta,\varepsilon \to 0}\ C = \b(\origin,\|x_j\|) \backslash \b(\origin,\|x_i\|)$ are distributed uniformly at random independent of each other. In other words, they form a uniform Binomial Point Process (BPP) over the annular region $\b(\origin,\|x_j\|) \backslash \b(\origin,\|x_i\|)$. By letting $i=0$, $j=L$, and defining $x_0 \triangleq \origin$ (i.e., $\|x_0\| = 0$), the BSs inside the circular region $\b(\origin,\|x_L\|)$ are shown to form a uniform BPP, thus proving Lemma~\ref{Lemma:CircularBPP}. By letting $i=1$ and $j=L$, Lemma~\ref{Lemma:AnnularBPP} is proved.

\subsection{Proof of Lemma~\ref{Lemma:FRa1}}\label{Proof:FRa1}
Let $\xa$ be an active BS that is uniformly distributed inside the circular region $\b(o, R_L)$. Then, for $\Omega \geq 1$,
\begin{align*}
	&F_{\Ra_1|R_L,\Omega}(r | R_L, \Omega)
	%= \P(\Ra_1 \leq r | R_L, \Omega) 
	= 1 - \P(\Ra_1 > r | R_L, \Omega) \\
	&= 1 - \P(\min\{ \xa : \|\xa\| < R_L \} > r | R_L, \Omega) \\
	&\stackrel{(a)}{=} 1 - \prod_{\{ \xa : \|\xa\| < R_L \}} \P(\|\xa\| > r | R_L) \\
	&\stackrel{(b)}{=} 1 - \prod_{\{ \xa : \|\xa\| < R_L \}} \frac{\pi R_L^2 - \pi r^2}{\pi R_L^2}
	= 1 - \left(\frac{R_L^2 - r^2}{R_L^2}\right)^\Omega,
\end{align*}
where $(a)$ and $(b)$ follow from Lemma~\ref{Lemma:CircularBPP}.

\subsection{Proof of Theorem~\ref{Theorem:pl_ubound}}\label{Proof:pl_ubound}
For $\Omega \geq 1$,
\begin{align*}
\pl &= \P\left( \frac{P R_L^{-\alpha}}{\sum_{i=1}^\infty P \Ra_i^{-\alpha} - P R_L^{-\alpha}} \geq \frac{\threshold}{\pg} \right) \\
&\leq \P\left( \frac{P R_L^{-\alpha}}{\sum_{i=1}^{\Omega+1} P \Ra_i^{-\alpha} - P R_L^{-\alpha}} \geq \frac{\threshold}{\pg} \right) \\
&\stackrel{(a)}{=}\P\left( \frac{R_L^{-\alpha}}{\Ra_1^{-\alpha} + \sum_{i=2}^{\Omega} \Ra_i^{-\alpha}} \geq \frac{\threshold}{\pg} \right) \\
&\stackrel{(b)}{\leq} \P\left( \frac{R_L^{-\alpha}}{\Ra_1^{-\alpha} + (\Omega-1) R_L^{-\alpha}} \geq\frac{\threshold}{\pg} \right) \\
%&= \P\left( \Ra_1^{-\alpha} \leq R_L^{-\alpha}\left(\frac{\pg}{\threshold} - (\Omega-1)\right) \right)
&= \P\left(\!\Ra_1 \geq R_L\left(\frac{\pg}{\threshold} - (\Omega-1)\right)^{\!\!-\frac{1}{\alpha}} \right)\\
&\stackrel{(c)}{=} \P\left(\!\min\{ \|\xa\|\!:\!\|\xa\| < R_L \} \geq R_L\left(\frac{\pg}{\threshold} - (\Omega-1)\right)^{\!\!-\frac{1}{\alpha}} \right) \\
%&= \E_{\Omega} \left[ \E_{R_L}\left [\prod_{\{\xa : \|\xa\| < R_L \}} \P\left( \|\xa\| \geq R_L\left(\frac{\pg}{\threshold} - (\Omega-1)\right)^{-\frac{1}{\alpha}} \middle\vert R_L, \Omega \right)\right ] \right ] \\
%&= \E_{\Omega} \left[ \E_{R_L} \left [ \prod_{\{\xa : \|\xa\| < R_L \}} \frac{\pi R_L^2 - \pi \left (R_L\left(\frac{\pg}{\threshold} - (\Omega-1)\right)^{-\frac{1}{\alpha}}\right )^2}{\pi R_L^2} \right ] \right ] \indicator\left(\pg^{-1}\threshold \leq \Omega^{-1}\right) \\
%&= \E_{\Omega} \left[ \E_{R_L}\left [\left(1-\left(\frac{\pg}{\threshold} - (\Omega-1)\right)^{-\frac{2}{\alpha}} \right)^{\Omega}\right ] \right ] \indicator\left(\pg^{-1}\threshold \leq \Omega^{-1}\right) \\
&\stackrel{(d)}{=} \E_{\Omega} \left[ \left(1-\left(\frac{\pg}{\threshold} - (\Omega-1)\right)^{-\frac{2}{\alpha}} \right)^{\Omega} \indicator\left(\frac{\threshold}{\pg} \leq \Omega^{-1}\right)\right] ,% \\
%&= \sum_{\omega=0}^{\lfloor \pg/\beta \rfloor} \left(1-\left(\frac{\pg}{\threshold} - (\omega-1)\right)^{-\frac{2}{\alpha}} \right)^{\omega} \binom{L-1}{\omega} p^\omega (1-p)^{L-1-\omega}.
\end{align*}
where $(a)$ follows from $\Ra_{\Omega+1} \triangleq R_L$ and pulling the dominant interferer out of the sum, $(b)$ from the fact that $\Ra_i \leq R_L$ for $i \leq \Omega$, $(c)$ since $\Ra_1$ is the smallest among all interferers, and $(d)$ from the fact that given $R_L$, all the BSs are uniformly distributed in the circle of radius $R_L$, followed by some algebraic manipulations. The final result is obtained by deconditioning over $\Omega \geq 1$, resulting in an expression which also trivially valid for the $\Omega=0$ case.

\subsection{Proof of Lemma~\ref{Lemma:EI1}}\label{Proof:EI1}
Let $\xa$ be an active BS inside the annular region $\b(o, R_L) \backslash \b(o, \Ra_1)$. Then,
\begin{align*}
&\E[P \|\xa\|^{-\alpha} | \Ra_1, R_L] \stackrel{(a)}{=} \int_{\Ra_1}^{R_L} P r^{-\alpha} \frac{2 r}{R_L^2 - \Ra_1^2}\ \d r \\
&= \frac{2P}{R_L^2 - \Ra_1^2} \int_{\Ra_1}^{R_L} r^{1-\alpha} \d r = \frac{2P}{R_L^2 - \Ra_1^2} \left( \frac{r^{2-\alpha}}{2-\alpha} \biggr|_{r = \Ra_1}^{R_L} \right)\\
&= \frac{2P}{2-\alpha}\cdot \frac{R_L^{2-\alpha} - \Ra_1^{2-\alpha}}{R_L^2 - \Ra_1^2},
\end{align*}
where $(a)$ follows from the fact that $\xa$ is uniformly distributed inside the annular region $\b(o, R_L) \backslash \b(o, \Ra_1)$ (proved in Lemma~\ref{Lemma:AnnularBPP}). Lastly, the mean of the sum of the interference from all interferers located inside the annular region $\b(o, R_L) \backslash \b(o, \Ra_1)$ is simply the sum of their individual means, from which the result follows.

\subsection{Proof of Lemma~\ref{Lemma:EI2}}\label{Proof:EI2}
In order to accommodate partial network loading, $q$ is the fraction of the far-off BSs transmitting throughout the localization procedure, which can be modeled using an equivalent thinned PPP $\PPPb$ with density $q\pppai$. For $\alpha > 2$, $\E[\Ix|R_L] = $
\begin{align*}
&\E\left[\sum_{x \in \PPPb \backslash \b(\origin, R_L)} P\|x\|^{-\alpha} \right]
\stackrel{(a)}{=} q\pppai \int_{\b^c(\origin, R_L)} P \|x\|^{-\alpha} \d x \\
&= P q\pppai \int_0^{2\pi} \int_{R_L}^\infty r^{-\alpha} r\ \d r\ \d \theta
= 2 P \pi q\pppai \int_{R_L}^\infty r^{1-\alpha} \d r,
%&= \left.2 P \pi q\pppai \cdot \frac{r^{2-\alpha}}{2-\alpha} \right\vert_{r=R_L}^\infty \\
%= \frac{2 P \pi q\pppai}{\alpha-2} R_L^{2-\alpha},
\end{align*}
where $(a)$ follows from Campbell's theorem~\cite{Haenggi2013}, and the result follows by solving the final integral (assuming $\alpha>2$).

\subsection{Proof of Proposition~\ref{Proposition:pl_EI1+EI2approx_p=0}}\label{Proof:pl_EI1+EI2approx_p=0}
The $L$-localizability probability in this case is $\pl \vert_{\Omega = 0}=$
\begin{align*}
&\E_{R_L}\left[ \indicator\left( \frac{P R_L^{-\alpha}}{\E[\Ix|R_L]} \geq \frac{\threshold}{\pg} \right) \right]
= \E_{R_L}\left[ \indicator\left( \frac{\pg}{\beta} \geq \frac{2 \pi q \pppai}{\alpha-2} R_L^2 \right) \right] \\
&= \E_{R_L}\left[ \indicator\left( \frac{\alpha-2}{2 \pi q \pppai \threshold / \pg} \geq R_L^2 \right) \right]
%= \E_{R_L}\left[ \indicator\left( R_L \leq \sqrt{\frac{\alpha-2}{2 \pi q \pppai \threshold / \pg}} \right) \right] \\
= \P\left( R_L \leq \sqrt{\frac{\alpha-2}{2 \pi q \pppai \threshold / \pg}} \right) \\
&\stackrel{(a)}{=} 1 - \sum_{\ell=0}^{L-1} e^{-\frac{\alpha-2}{2 q \threshold / \pg}} \frac{(\frac{\alpha-2}{2 q \threshold/\pg})^\ell}{\ell!},
\end{align*}
where $(a)$ follows from the fact that the previous expression is simply the probability that \emph{at least} $L$ BSs lie inside the region $\b\left(\origin, \sqrt{\frac{\alpha-2}{2 \pi q \pppai \threshold / \pg}} \right)$.

\subsection{Proof of Theorem~\ref{Theorem:pl_EI1+EI2approx}}\label{Proof:pl_EI1+EI2approx}
The general $L$-localizability probability is
\begin{align}
\pl
&= \P\left(\SIR_L \geq \frac{\threshold}{\pg}\right) \notag \\
&= \E_{\Omega} \left[ \E_{R_L} \left[ \E_{\Ra_1} \left[ \indicator\left( \SIR_L \geq \frac{\threshold}{\pg} \middle\vert \Ra_1, R_L, \Omega \right) \middle\vert R_L, \Omega \right] \right] \right].
\label{eq:Thm2_intermediate_HD1}
\end{align}
The $\SIR_L$ term in the above expression is
\begin{align}
&\SIR_L = \frac{P R_L^{-\alpha}}{P \Ra_1^{-\alpha} + \E[\Ia | \Ra_1, R_L, \Omega] + \E[\Ix | R_L]} \\
&= \frac{P R_L^{-\alpha}}{P \Ra_1^{-\alpha} + \frac{2 P (\Omega-1)}{2-\alpha}\cdot \frac{R_L^{2-\alpha} - \Ra_1^{2-\alpha}}{R_L^2 - \Ra_1^2} + \frac{2 P \pi q\pppai}{\alpha-2}R_L^{2-\alpha}}. \label{eq:Thm2_intermediate_HD2} 
\end{align}
The result follows by substituting this expression for $\SIR_L$ back in \eqref{eq:Thm2_intermediate_HD1}, followed by integrating over the densities of $\Ra_1$, $R_L$, and $\Omega \geq 1$, and treating the $\Omega=0$ case separately using Proposition~\ref{Proposition:pl_EI1+EI2approx_p=0}.

\subsection{Proof of Corollary~\ref{Corollary:pl_EI1+EI2approx_singleintegral}}\label{Proof:pl_EI1+EI2approx_singleintegral}
As in the Proof of Theorem~\ref{Theorem:pl_EI1+EI2approx} above, $\pl$ can be expressed in terms of $\SIR_L$ as \eqref{eq:Thm2_intermediate_HD1}, where $\SIR_L$ is given by \eqref{eq:Thm2_intermediate_HD2}. In this Corollary, we multiply the $\E[\Ix|R_L]$ term of $\SIR_L$ by $\E[\Ra_1^2]/\Ra_1^2 = 1/(\pi \lambda \Ra_1^2)$ to get a simpler approximation for $\SIR_L$, and consequently for $\pl$, as follows:
\begin{align}
\SIR_L
&= \frac{P R_L^{-\alpha}}{P \Ra_1^{-\alpha} + \frac{2 P (\Omega-1)}{2-\alpha}\cdot \frac{R_L^{2-\alpha} - \Ra_1^{2-\alpha}}{R_L^2 - \Ra_1^2} + \frac{2 P \pi q\pppai}{\alpha-2}R_L^{2-\alpha}\frac{\E[\Ra_1^2]}{\Ra_1^2}} \notag \\
&\stackrel{(a)}{=}  \frac{1}{X^{\alpha} + \frac{2(\Omega-1)}{2-\alpha}\cdot \frac{X^{2} - X^{\alpha}}{X^2 - 1} + \frac{2q}{\alpha-2}X^2} ,\label{Eq:pl_EI1+EI2approx_singleintegral_appendix}
\end{align}
where $(a)$ follows by substituting $R_L/\Ra_1 \rightarrow X$. The distribution of $X$ can be readily obtained from \eqref{Eq:F_Ra_1_given_R_L} as $F_X(x) = (1-1/x^2)^\Omega$. Note that after this substitution, the dependence of $\SIR_L$ on $\Ra_1$ and $R_L$ is completely captured through $X$. The final expression is now obtained by substituting \eqref{Eq:pl_EI1+EI2approx_singleintegral_appendix} in \eqref{eq:Thm2_intermediate_HD1}, followed by integrating over the support of $X$ and subsequently summing over the probability mass of $\Omega$, excluding $\Omega=0$ which is treated separately using Proposition~\ref{Proposition:pl_EI1+EI2approx_p=0}. The key here is that instead of the {\em double} integration over $\Ra_1$ and $R_L$, {\em single} integration over their ratio $X$ suffices in this case.

\subsection{Proof of Corollary~\ref{Corollary:pl_EI1+EI2approx_a=4}}\label{Proof:pl_EI1+EI2approx_a=4}
As in the proof of Corollary~\ref{Corollary:pl_EI1+EI2approx_singleintegral}, we start with the expression of $\pl$ given by \eqref{eq:Thm2_intermediate_HD1} in the proof of Theorem~\ref{Theorem:pl_EI1+EI2approx}. The $\SIR_L$ term appearing in the expression of $\pl$ is given by \eqref{eq:Thm2_intermediate_HD2}. In this Corollary, we simplify the $\SIR_L$ expression for the case of $\alpha=4$. Starting with the $\SIR_L$ expression given by \eqref{eq:Thm2_intermediate_HD2}:
\begin{align}
\SIR_L
%&= \frac{P R_L^{-\alpha}}{P \Ra_1^{-\alpha} + \E[\Ia | \Ra_1, R_L] + \E[\Ix | R_L] } \notag \\
&=  \frac{P R_L^{-\alpha}}{P \Ra_1^{-\alpha} + \frac{2 P (\Omega-1)}{2-\alpha}\cdot \frac{R_L^{2-\alpha} - \Ra_1^{2-\alpha}}{R_L^2 - \Ra_1^2} + \frac{2 P \pi \pppai}{\alpha-2}R_L^{2-\alpha} } \nonumber \\
%&=  \frac{P}{P \left(\frac{R_L}{\Ra_1}\right)^{\alpha} + \frac{2 P (\Omega-1)}{2-\alpha}\cdot \frac{R_L^{2} - \Ra_1^{2-\alpha}R_L^{\alpha}}{R_L^2 - \Ra_1^2} + \frac{2 P \pi \pppai}{\alpha-2}R_L^2} \vert_{\alpha = 4} \\
&= \frac{1}{\left(\frac{R_L}{\Ra_1}\right)^{\alpha} + \frac{2(\Omega-1)}{2-\alpha}\cdot \frac{\left(\frac{R_L}{\Ra_1}\right)^{2} - \left(\frac{R_L}{\Ra_1}\right)^{\alpha}}{\left(\frac{R_L}{\Ra_1}\right)^2 - 1} + \frac{2 \pi \pppai}{\alpha-2}R_L^2} \notag \\
%&= \frac{1}{X^{\alpha} + \frac{2(\Omega-1)}{2-\alpha}\cdot \frac{X^{2} - X^{\alpha}}{X^2 - 1} + \frac{2 \pi \pppai}{\alpha-2}R_L^2} \vert_{\alpha = 4} \\
%&= \frac{1}{X^{4} + \frac{2(\Omega-1)}{2-4}\cdot \frac{X^{2} - X^{4}}{X^2 - 1} + \pi \pppai R_L^2}  \\
%&= \frac{1}{X^{4} + \frac{2(\Omega-1)}{2}\cdot \frac{X^{2}\left(X^2 - 1\right)}{X^2 - 1} + \pi  \pppai R_L^2} \\
&\stackrel{(a)}{=} \frac{1}{X^{4} + (\Omega-1) X^2 + \pi \pppai R_L^2},
\label{eq:SINR_L_hd}
\end{align}
where $(a)$ follows by defining $X = \frac{R_L}{\Ra_1}$ and substituting $\alpha=4$. Using \eqref{eq:SINR_L_hd}, we now simplify $\SIR_L \ge \threshold/\pg$ which will then be used to derive $\pl$ from \eqref{eq:Thm2_intermediate_HD1}. Defining $Y = X^2$ we have
\begin{align*}
&\SIR_L \ge \threshold/\pg %\Rightarrow  Y^2 + (\Omega-1) Y + \pi \pppai R_L^2 \le \pg/\threshold  \\
\stackrel{(a)}{\Rightarrow} Y^2 + (\Omega-1) Y \le \kappa^{-1}\\
&\Rightarrow \left(Y + \frac{(\Omega-1)}{2} \right)^2 \le \kappa^{-1} + \frac{(\Omega-1)^2}{4} \\
&\stackrel{(b)}{\Rightarrow} 0 \le Y \le \sqrt{\kappa^{-1} + \frac{(\Omega-1)^2}{4}} - \frac{(\Omega-1)}{2}\\
&\stackrel{(c)}{\Rightarrow}  X^2 \le \sqrt{\kappa^{-1} + \frac{(\Omega-1)^2}{4}} - \frac{(\Omega-1)}{2} \\
&\stackrel{(d)}{\Rightarrow} 1 \le X \le \sqrt{\sqrt{\kappa^{-1} + \frac{(\Omega-1)^2}{4}} - \frac{(\Omega-1)}{2}}\\
&\Rightarrow \frac{R_L}{\sqrt{\sqrt{\kappa^{-1} + \frac{(\Omega-1)^2}{4}} - \frac{(\Omega-1)}{2}}} \le \Ra_1 \le R_L,
\end{align*}
where $\kappa^{-1} = \pg/\threshold - \pi \pppai R_L^2$ in $(a)$, $(b)$ follows from the fact that $Y \geq 1$, and $(c)$ from $Y = X^2$. Note that $(b)$ and $(c)$ require $\kappa^{-1} \geq -\frac{(\Omega-1)^2}{4}$. Step $(d)$ follows from $X\geq 1$. The earlier condition on $\kappa^{-1}$ is replaced by a more strict condition $\kappa^{-1} \geq \Omega$ in $(d)$. Substituting this back in \eqref{eq:Thm2_intermediate_HD1} and doing some algebraic manipulations, we get
\begin{align*}
\pl
&= \E_{\Omega} \left[ \E_{R_L} \left[\left(\frac{R_L^2 - \frac{R_L^2}{\sqrt{\kappa^{-1} + \frac{(\Omega-1)^2}{4}} - \frac{(\Omega-1)}{2}}}{R_L^2}\right)^{\!\!\!\Omega\,} \right] \right] \\
&= \E_{\Omega} \left[ \left(1 - \frac{1}{\sqrt{\kappa^{-1} + \frac{(\Omega-1)^2}{4}} - \frac{(\Omega-1)}{2}}\right)^{\!\!\!\Omega\,} \right]
\end{align*}
from which the result follows by simply deconditioning over $\Omega$. Due to $\kappa^{-1} \geq \Omega$, the integration limits are from $0$ to $\sqrt{\frac{\pg/\threshold-\omega}{\pi \pppai}}$.

%%%%%%%%%
% Bibliography %
%%%%%%%%%

\bibliographystyle{IEEEtran}
\bibliography{Paper_TW_Feb_15_0254_R2,IEEEabrv}

\end{document}

%% file: notation.tex
\def\ncalA{{\mathcal{A}}}
\def\ncalB{{\mathcal{B}}}

\def\ncalH{{\mathcal{H}}}
\def\ncalI{{\mathcal{I}}}

\def\ncalS{{\mathcal{S}}}

\def\argmax{\operatorname*{arg~max}}
\def\E{\mathbb{E}}
\def\P{\mathbb{P}}
\def\b{\mathbbmss{b}}   % ball (requires package bbm)
\def\R{\mathbb{R}}

\def\indicator{{\mathbb{1}}}

\def\d{\mathrm{d}}

\def\th{{^\text{th}}}

\def\pl{\mathtt{P_L}}
\def\plk{\mathtt{P_L^K}}
\def\PPPa{{\Phi}}
\def\pppa{{\phi}}
\def\pppai{{\lambda}}

\def\PPPb{{\Psi}}

\def\Ia{{\ncalI_1}}
\def\Ix{{\ncalI_2}}
\def\N{{\sigma^2}}
\def\threshold{{\beta}}
\def\SINR{{\sf SINR}}

\def\SIR{{\sf SIR}}

\def\xa{{\hat{x}}}
\def\Ra{{\hat{R}}}
\def\ra{{\hat{r}}}
\def\origin{{o}}
\def\pg{{\gamma}}
\def\btheta{{\boldsymbol{\theta}}}
\def\pOmega{{f_{\Omega}}}
\def\pomega{{\pOmega(\omega)}}
\def\omegasum{{\sum_{\omega=1}^{L-1}\pomega}}
\def\GammaL{{(L-1)!}}